\RequirePackage{tikzpagenodes} 
\documentclass[default,iicol]{sn-jnl}% Default with double column layout

\jyear{2021}%

\usepackage{enumerate}
\usepackage{float} %for placing figure at precise location, using \begin{figure*}[H].

\usepackage{graphicx} % Required for the inclusion of images
\usepackage{amsmath} % Required for some math elements 
\usepackage{amssymb}
\usepackage{amsthm}
\usepackage[utf8]{inputenc}
\usepackage[english]{babel}
\setcitestyle{square}
\edef\normalE{\the\mathcode`E}
\begingroup\lccode`~=`E \lowercase{\endgroup
  \def~}{\mathbb{\mathchar\normalE}}
\mathcode`E="8000

\theoremstyle{thmstyleone}%
\newtheorem{theorem}{Theorem}[section]

\newtheorem{lemma}[theorem]{Lemma}
\newtheorem{prop}[theorem]{Proposition}%[theorem]

\theoremstyle{thmstyletwo}%
\newtheorem{remark}{Remark}%

\theoremstyle{thmstylethree}%
\newtheorem{definition}{Definition}%

\begin{document}

\title[Article Title]{Morphological Sampling Theorem and its Extension to Grey-value Images}


\author*[1]{\fnm{Vivek} \sur{Sridhar}}\email{sridharvivek95@b-tu.de}

\author[1]{\fnm{Michael} \sur{Breu{\ss}}}\email{breuss@b-tu.de}


\affil*[1]{\orgdiv{Institute for Mathematics}, \orgname{BTU Cottbus-Senftenberg}, \orgaddress{\street{Platz der Deutschen Einheit 1}, \city{Cottbus}, \postcode{03046}, \state{Brandenburg}, \country{Germany}}}

%\affil[2]{\orgdiv{Department}, \orgname{Organization}, \orgaddress{\street{Street}, \city{City}, \postcode{10587}, \state{State}, \country{Country}}}

%\affil[3]{\orgdiv{Department}, \orgname{Organization}, \orgaddress{\street{Street}, \city{City}, \postcode{610101}, \state{State}, \country{Country}}}


\abstract{Sampling is a basic operation in image processing. 
In classic literature, 
a morphological sampling theorem has been established, which shows 
how sampling interacts by morphological operations with image reconstruction. 
Many aspects of morphological sampling have been investigated for binary 
images, but only some of them have been explored for grey-value imagery. 

With this paper, we make a step towards completion of this open matter. 
By relying on the umbra notion, we show how to transfer classic theorems 
in binary morphology about the interaction of sampling with 
the fundamental morphological operations
dilation, erosion, opening and closing, to the grey-value setting. 
In doing this we also extend the theory relating the morphological operations
and corresponding reconstructions
%, 
%max-pooling and its corresponding reconstruction, 
to use of non-flat structuring elements. %Extension of Heijman's work
We illustrate the theoretical developments at hand of examples.
%We only give examples of one grey-value image.
}




\maketitle

\section{Introduction}\label{Sec:Intro}

Mathematical morphology is a very successful approach in image  processing, cf.\ \cite{Serra-Soille,Najman-Talbot,Roerdink-2011} for an account. Morphological filters make use of a so called structuring element (SE). The SE is characterised by shape, size and centre location. There are two types of SEs, flat and non-flat \cite{r1}. A flat SE defines a neighbourhood of the centre pixel where morphological operations take place.
A non-flat SE may additionally contain finite values used as additive offsets. 
The basic morphological operations are dilation and erosion. In a discrete setting as discussed in this work, these operations are realized by setting 
a pixel value to the maximum or minimum of the discrete image function within the SE centred upon it, respectively. The fundamental building blocks dilation and erosion
may be combined to many morphological processes of practical interest, 
like e.g.\ opening, closing or top hats. 

Sampling is a basic operation in signal and image processing. 
The celebrated Nyquist-Shannon sampling theorem relates the 
bandwidth of a continuous-scale signal
to its reconstruction via equidistant sampled values, cf.\ \cite{Shannon-1998}
for an account. 
Turning to morphological filters, the classic sampling theorem
has an analogon within the framework of discrete sets and lattices. 
In this setting, the proceeding is based on image reconstruction 
by using samples together with the standard morphological processes 
of dilation and erosion as well as their combinations;
see the classic works of Haralick and co-authors \cite{r2, Haralick88}
as well as previous developments in \cite{r1}.
In these works, sampling on the image grid was put in relation 
with image reconstruction via dilation and closing,
and formulated in \cite{r2} as the 
digital morphological sampling theorem.
%for both binary and (in a limited way as discussed below)
%grey-value imagery, respectively.

Considering the literature that followed the seminal work \cite{r2} on morphological sampling, we are not aware of further elaborations on sampling issues related to morphological filters.
To the best of our knowledge the results documented in \cite{r2} have been cited for giving a
theoretical basis for different developments, but they have not been continued at exactly that point. 
However, in the mentioned work several mathematical
assertions related to sampling and its interaction
with the basic morphological processes dilation, 
erosion opening and closing have been addressed only for 
binary images, while they have not been carried over to 
the setting of grey value imagery. More precisely, 
this open issue refers to the situation when filtering
morphologically in the sampled domain, i.e., making use
of a sampled account of a given image.
This is from a computational point of view an interesting
setting since the image dimension and thus the amount
of necessary filtering operations may be reduced by sampling
considerably.

In the same line of classic works, the relation between 
morphological operations and reconstruction of grey-value images 
was explored in the context of max-pooling in \cite{r3}. 
Let us note that the max-pooling operation as introduced 
in \cite{ZC-88} is often used in convolutional neural 
networks \cite{Goodfellow-et-al-2016}.

%Coming back to the classic digital morphological sampling theorem,
Let us elaborate a bit more at this point. 
While \cite{r3} presents relations between
operating morphologically before sampling and
morphologically operating in the sampled domain, 
there are several limitations by the proceeding in \cite{r3},
so that it represents an important step in the investigation
of morphological operations in context of sampling and reconstruction,
but it is not a complete theory. 
First, the theory in \cite{r3} is limited to a particular 
type of sampling, i.e.\ max-pooling. 
Max-pooling is the dilation by a square flat SE followed 
by sampling, i.e., dilation is used in a first 
step as a filter before the sampling step. 
%This bears some practical implications, e.g.\ noise is effectively kept by dilating an image. 
Reconstruction is also obtained in \cite{r3} by dilation. 
Since dilation is not an idempotent operation, it is seldom used 
directly as a filter in real world applications, and a corresponding 
sampling and reconstruction setting as in \cite{r3} bears 
considerable restrictions. 
Let us also note that the work \cite{r3} assumes that the 
SE is already a subset of the sampled domain, i.e., 
technically the SE is affected by and acts on only those 
pixels of the image which are sampled.
Finally, the work in \cite{r3} is limited to flat filters 
and SEs. The discussions are limited to the lattice-algebraic 
framework \cite{r4}, which lacks the tools to work with 
non-flat morphology. 

Let us also elaborate a little more on the latter aspect in order 
to clarify the nature of our proposed extensions upon aforementioned 
classic works. The lattice algebraic theory provides a rich framework to study 
mathematical morphology, see e.g.~\cite{r4, r5}.
Lattice theory examines morphological operations as transformations 
on the complete lattice group of images. The ordering within the lattice
is based on the inherent ordering, whether partial or total, of the 
pixel values, which, in the case of grey-value images, is the total 
order of the set of integers that make up the grey values. 
To summarize, lattice theory is largely based on \emph{tonal} relationship
between pixels. However, the lattice theory in itself lacks effective tools to deal 
with non-flat SEs and sampling.  In particular, it largely neglects 
the \emph{spatial} relationship between pixels of an image and the pixels of an SE 
as it moves across the image during morphological operations. The constraints 
imposed by lattice theory on the study of grey-value morphological sampling are 
apparent in \cite{r3}, where the authors attempt to employ this approach. 
Firstly, the filters and S.Es\ are restricted to be flat. Moreover, the SE is 
already in the sampled domain, that is, the action of SE on the pixels of the 
image which are sampled is unaffected by the image pixels which are not sampled.  

To overcome the aforementioned limitations, we employ the \emph{umbra formulation} of 
grey-value images. Umbra technique allows us to treat morphological operations 
on $N$-dimensional grey-value images (by flat or non-flat SEs) as binary morphological 
operations of their corresponding $(N+1)$-dimensional umbras \cite{r1}. 
Binary morphology is in turn founded on the \emph{spatial} arrangement of pixels. 
Its basic operations are set operations, where both the binary image and SE are treated 
as sets of positional vectors. Furthermore, \cite{r2} thoroughly examines the connection 
between binary morphology and sampling.

{\bf Our Contributions.}
In a previous conference paper \cite{sridhar-breuss-sampling-caip}, 
we have shown that 
it is possible to extend the work in \cite{r3} to non-flat 
filters and SEs using umbra formulation of morphological operations, 
and we have proposed a few results that can be derived from \cite{r2}. 
More precisely, we explored an alternative definition of
grey-scale opening and closing to prove reconstruction bounds
for the interaction of sampling with these operations.
In the current paper we build upon \cite{sridhar-breuss-sampling-caip}
and extend the classic work of Haralick et al.\ in \cite{r2} 
on digital morphological sampling of grey value images.
As the main point of our developments, we formulate and prove 
theorems relating morphological operations, sampling and 
image reconstruction by dilation and closing. Let us point out 
again clearly, that corresponding results have been derived 
in a relatively simple way for binary images in \cite{r2}, 
but they have not been extended up to now to grey-value imagery.
The theory we formulate is also used here to extend the work in 
\cite{r3} to non-flat morphology. In doing this, we give a theoretical 
foundation for specific uses of the max-pooling operation 
used in modern deep learning literature.
In total, compared to \cite{sridhar-breuss-sampling-caip},
we give a much more extensive account of the theoretical framework,
proving in this context several additional results.

We believe that especially the extension to non-flat morphology
may be an interesting point with respect to recent developments 
in incorporating %adaptive morphological filters %include citations
morphological layers in neural networks, 
see for instance \cite{Shen-2019, Franchi-2020, Kirszenberg-2021, Hermary-2022}. 
There the learned morphological filters within the layers
are usually non-flat. Thus we give a theoretical foundation
of this recent machine learning technique by the current paper.

{\bf Paper Organisation.} In the next section we briefly recall the
classic notions from mathematical morphology for binary as well 
as grey value images, see for instance
\cite{r1,Matheron-75,Serra-82,Serra-88,Soille-2003}
for a corresponding account of the field and the 
basic notions underlying our work.
In addition, we will briefly recall the digital morphological 
sampling theorem in the
binary and grey value setting, respectively. 
The third section contains the main part of our new results. 
In the fourth section, we use the developed results to extend 
the theory in \cite{r3} to non-flat structuring elements.
We visualize the meaning of theoretical developments by 
some experiments within the text.
The paper is finished by concluding remarks.

\section{Morphological Operations}
\label{Sec:Fundamentals}

As indicated we now recall formal definitions and 
some fundamental properties of morphological operations
that help to assess the later developments.

\subsection{Morphological Notions for Binary Images}

Let $E$ denote the set of integers used to index the rows and column of the image. $E^N$ is a $N$-tuple of $E$.  A (two dimensional)binary image $A$ is a subset of $E^2$.  That is, if a vector $x\in A\subseteq E$, then the position at  $x$ is a \textit{white} dot, where the default background is \textit{black}. For sake of generality, we consider the image as $A\subseteq E^N$, $N\in \mathbb{N}$ \cite{r1}.

\begin{definition}{\textbf{Translation, Dilation and Erosion, Reflection, Duality.}}
\label{def:bdilation}
Let $A$, $B$ be subsets of $E^N$. For $x \in E^N$, the translation of $A$ by $x$ is
written as $(A)_x $ $=$ $\{ c\in E^N \vert c=a+x \text{ for some }  a\in A \} $.
The dilation of $A$ by $B$ is defined as
\begin{align}
\begin{split}
  A\oplus B &= \{ c\in E^N  \vert  c=a+b \text{ for some }  a\in A, b\in B\} \\
  &=  \bigcup_{b\in B} (A)_b
\end{split}
  \label{dilation}
\end{align}
The erosion of set $A$ by $B$ is defined
as
\begin{align}
\begin{split}
  A\ominus B &= \{x  \vert  x+b\in A \text{ for each } b\in B\} \\
  &=  \{ x\in E^N  \vert  (B)_x \subseteq A \}\\
  &=  \bigcap_{b\in B} (A)_{-b}
\end{split}
  \label{erosion}
\end{align}
In addition, the reflection of a set $B$ is denoted by $\breve{B} = \{x  \vert  \text{ for some } b\in B \text{, } x=-b \} $.
Moreover, it holds duality in the sense $(A\ominus B)^c = A^c \oplus \breve{B} $.
\end{definition}
Opening and Closing as described below can be employed to erase image details smaller than the structuring element without
distorting unsuppressed geometric features, see e.g.\ \cite{r1}. They can easily be generalised to the grey value setting.

\begin{definition}{\textbf{Opening and Closing, Duality of Opening and Closing.}}
\label{binary-opening-closing}
The opening of $B \subseteq E^N$ by structuring element $K$ is denoted by $B\circ K$ and is
defined as $B\circ K= (B\ominus K)\oplus K$. Analogously, opening is denoted as
$B\bullet K= (B\oplus K)\ominus K$. The operations are dual i.e.\ $(A\bullet B)^c =A^c \circ B$.
\end{definition}
We note that there exist the following alternative definitions of opening and closing which
are useful in proofs of various results:
\begin{align}
\begin{split}
  A\circ B  &= 
  \{ x\in A  \vert  \text{ for some } y \text{, } x\in B_y \subseteq A \} \\
  &=  \bigcup _{\{y \vert B_y\subseteq A \}} B_y
\end{split}
  \label{alt-open}
\end{align}
and
\begin{align}
\begin{split}
  A\bullet B &= 
  \{ x  \vert  x \in \breve{B_y} \text{ implies } \breve{B_y} \cap A \neq \emptyset \}
 \end{split}
  \label{alt-close}
\end{align}
Let us now give some comments on the meaning of the binary digital morphological sampling theorem at hand of an example, see Figure \ref{fig:KBS}. As the corresponding sampling theorem will be recalled in detail for the grey-value setting, we refrain from giving a more detailed exposition here.

%%%%%%%%%%%%%%%%%%%%%%%%%%%%%%%%%%%%%%%%%%%%%%%%%%%%%%%%%%%%%%%%%%%%%%%%%%%%%%%%%%%%%%%%%%%%%%%%%
% Figure KBS
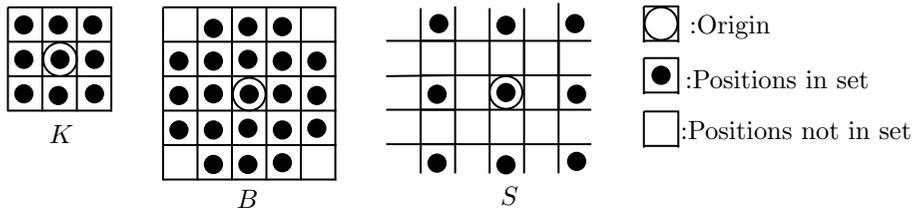
\begin{figure*}[!htbp]
\centering

\tikzset{every picture/.style={line width=0.75pt}} %set default line width to 0.75pt       

\begin{tikzpicture}[x=0.65pt,y=0.65pt,yscale=-1,xscale=1]
%uncomment if require: \path (0,140.484375); %set diagram left start at 0, and has height of 140.484375

%Shape: Grid [id:dp9597332067298479]
\draw  [draw opacity=0] (11,10) -- (71.5,10) -- (71.5,70.48) -- (11,70.48) -- cycle ; \draw   (11,10) -- (11,70.48)(31,10) -- (31,70.48)(51,10) -- (51,70.48)(71,10) -- (71,70.48) ; \draw   (11,10) -- (71.5,10)(11,30) -- (71.5,30)(11,50) -- (71.5,50)(11,70) -- (71.5,70) ; \draw    ;
%Shape: Grid [id:dp1651748896103642]
\draw  [draw opacity=0] (101,10) -- (201.5,10) -- (201.5,110.48) -- (101,110.48) -- cycle ; \draw   (101,10) -- (101,110.48)(121,10) -- (121,110.48)(141,10) -- (141,110.48)(161,10) -- (161,110.48)(181,10) -- (181,110.48)(201,10) -- (201,110.48) ; \draw   (101,10) -- (201.5,10)(101,30) -- (201.5,30)(101,50) -- (201.5,50)(101,70) -- (201.5,70)(101,90) -- (201.5,90)(101,110) -- (201.5,110) ; \draw    ;
%Shape: Grid [id:dp32391297773633343]
\draw  [draw opacity=0] (250.5,29.48) -- (349.5,29.48) -- (349.5,106.48) -- (250.5,106.48) -- cycle ; \draw   (250.5,29.48) -- (250.5,106.48)(270.5,29.48) -- (270.5,106.48)(290.5,29.48) -- (290.5,106.48)(310.5,29.48) -- (310.5,106.48)(330.5,29.48) -- (330.5,106.48) ; \draw   (250.5,29.48) -- (349.5,29.48)(250.5,49.48) -- (349.5,49.48)(250.5,69.48) -- (349.5,69.48)(250.5,89.48) -- (349.5,89.48) ; \draw    ;
%Straight Lines [id:da49576989034328145]
\draw    (270.5,9.48) -- (270.5,29.48) ;

%Straight Lines [id:da4784950834218822]
\draw    (290.5,10.48) -- (290.5,49.48) ;

%Straight Lines [id:da6509997279912634]
\draw    (310.5,10.48) -- (310.5,49.48) ;

%Straight Lines [id:da004332343017760643]
\draw    (330.5,7.48) -- (330.5,69.48) ;

%Straight Lines [id:da5144308785012996]
\draw    (330.5,30.48) -- (330.5,69.48) ;

%Straight Lines [id:da46356437372208936]
\draw    (250.5,10.48) -- (250.5,89.48) ;

%Straight Lines [id:da8249473844061388]
\draw    (230.5,89.48) -- (270.5,89.48) ;

%Straight Lines [id:da2346998085119294]
\draw    (230.5,69.48) -- (270.5,69.48) ;

%Straight Lines [id:da07047637134982176]
\draw    (230.5,29.48) -- (270.5,29.48) ;

%Straight Lines [id:da8947555055780307]
\draw    (230.5,49.98) -- (270.5,49.48) ;

%Shape: Circle [id:dp7724242829636028]
\draw   (31.5,40.23) .. controls (31.5,34.85) and (35.86,30.48) .. (41.24,30.48) .. controls (46.62,30.48) and (50.98,34.85) .. (50.98,40.23) .. controls (50.98,45.61) and (46.62,49.97) .. (41.24,49.97) .. controls (35.86,49.97) and (31.5,45.61) .. (31.5,40.23) -- cycle ;
%Shape: Circle [id:dp9982013591868288]
\draw   (141.5,60.47) .. controls (141.5,55.22) and (145.75,50.97) .. (151,50.97) .. controls (156.25,50.97) and (160.5,55.22) .. (160.5,60.47) .. controls (160.5,65.72) and (156.25,69.97) .. (151,69.97) .. controls (145.75,69.97) and (141.5,65.72) .. (141.5,60.47) -- cycle ;
%Shape: Circle [id:dp44050452510756455]
\draw   (290.5,59.47) .. controls (290.5,54.22) and (294.75,49.97) .. (300,49.97) .. controls (305.25,49.97) and (309.5,54.22) .. (309.5,59.47) .. controls (309.5,64.72) and (305.25,68.97) .. (300,68.97) .. controls (294.75,68.97) and (290.5,64.72) .. (290.5,59.47) -- cycle ;
%Shape: Circle [id:dp9817530157897594]
\draw  [fill={rgb, 255:red, 0; green, 0; blue, 0 }  ,fill opacity=1 ] (36.24,40.23) .. controls (36.24,37.47) and (38.48,35.23) .. (41.24,35.23) .. controls (44,35.23) and (46.24,37.47) .. (46.24,40.23) .. controls (46.24,42.99) and (44,45.23) .. (41.24,45.23) .. controls (38.48,45.23) and (36.24,42.99) .. (36.24,40.23) -- cycle ;
%Shape: Circle [id:dp6578755971583434]
\draw  [fill={rgb, 255:red, 0; green, 0; blue, 0 }  ,fill opacity=1 ] (56.24,60.23) .. controls (56.24,57.47) and (58.48,55.23) .. (61.24,55.23) .. controls (64,55.23) and (66.24,57.47) .. (66.24,60.23) .. controls (66.24,62.99) and (64,65.23) .. (61.24,65.23) .. controls (58.48,65.23) and (56.24,62.99) .. (56.24,60.23) -- cycle ;
%Shape: Circle [id:dp09219418068621321]
\draw  [fill={rgb, 255:red, 0; green, 0; blue, 0 }  ,fill opacity=1 ] (55.24,20.23) .. controls (55.24,17.47) and (57.48,15.23) .. (60.24,15.23) .. controls (63,15.23) and (65.24,17.47) .. (65.24,20.23) .. controls (65.24,22.99) and (63,25.23) .. (60.24,25.23) .. controls (57.48,25.23) and (55.24,22.99) .. (55.24,20.23) -- cycle ;
%Shape: Circle [id:dp6217833476594274]
\draw  [fill={rgb, 255:red, 0; green, 0; blue, 0 }  ,fill opacity=1 ] (56.24,40.23) .. controls (56.24,37.47) and (58.48,35.23) .. (61.24,35.23) .. controls (64,35.23) and (66.24,37.47) .. (66.24,40.23) .. controls (66.24,42.99) and (64,45.23) .. (61.24,45.23) .. controls (58.48,45.23) and (56.24,42.99) .. (56.24,40.23) -- cycle ;
%Shape: Circle [id:dp9773503909509569]
\draw  [fill={rgb, 255:red, 0; green, 0; blue, 0 }  ,fill opacity=1 ] (35.24,61.23) .. controls (35.24,58.47) and (37.48,56.23) .. (40.24,56.23) .. controls (43,56.23) and (45.24,58.47) .. (45.24,61.23) .. controls (45.24,63.99) and (43,66.23) .. (40.24,66.23) .. controls (37.48,66.23) and (35.24,63.99) .. (35.24,61.23) -- cycle ;
%Shape: Circle [id:dp16218573888938703]
\draw  [fill={rgb, 255:red, 0; green, 0; blue, 0 }  ,fill opacity=1 ] (15.24,60.23) .. controls (15.24,57.47) and (17.48,55.23) .. (20.24,55.23) .. controls (23,55.23) and (25.24,57.47) .. (25.24,60.23) .. controls (25.24,62.99) and (23,65.23) .. (20.24,65.23) .. controls (17.48,65.23) and (15.24,62.99) .. (15.24,60.23) -- cycle ;
%Shape: Circle [id:dp10821121157814861]
\draw  [fill={rgb, 255:red, 0; green, 0; blue, 0 }  ,fill opacity=1 ] (15.24,40.23) .. controls (15.24,37.47) and (17.48,35.23) .. (20.24,35.23) .. controls (23,35.23) and (25.24,37.47) .. (25.24,40.23) .. controls (25.24,42.99) and (23,45.23) .. (20.24,45.23) .. controls (17.48,45.23) and (15.24,42.99) .. (15.24,40.23) -- cycle ;
%Shape: Circle [id:dp07534596880405409]
\draw  [fill={rgb, 255:red, 0; green, 0; blue, 0 }  ,fill opacity=1 ] (35.24,20.23) .. controls (35.24,17.47) and (37.48,15.23) .. (40.24,15.23) .. controls (43,15.23) and (45.24,17.47) .. (45.24,20.23) .. controls (45.24,22.99) and (43,25.23) .. (40.24,25.23) .. controls (37.48,25.23) and (35.24,22.99) .. (35.24,20.23) -- cycle ;
%Shape: Circle [id:dp36781619397209586]
\draw  [fill={rgb, 255:red, 0; green, 0; blue, 0 }  ,fill opacity=1 ] (15.24,20.23) .. controls (15.24,17.47) and (17.48,15.23) .. (20.24,15.23) .. controls (23,15.23) and (25.24,17.47) .. (25.24,20.23) .. controls (25.24,22.99) and (23,25.23) .. (20.24,25.23) .. controls (17.48,25.23) and (15.24,22.99) .. (15.24,20.23) -- cycle ;
%Shape: Circle [id:dp052376465927621574]
\draw  [fill={rgb, 255:red, 0; green, 0; blue, 0 }  ,fill opacity=1 ] (145.24,21.23) .. controls (145.24,18.47) and (147.48,16.23) .. (150.24,16.23) .. controls (153,16.23) and (155.24,18.47) .. (155.24,21.23) .. controls (155.24,23.99) and (153,26.23) .. (150.24,26.23) .. controls (147.48,26.23) and (145.24,23.99) .. (145.24,21.23) -- cycle ;
%Shape: Circle [id:dp23572897269268878]
\draw  [fill={rgb, 255:red, 0; green, 0; blue, 0 }  ,fill opacity=1 ] (105.24,41.23) .. controls (105.24,38.47) and (107.48,36.23) .. (110.24,36.23) .. controls (113,36.23) and (115.24,38.47) .. (115.24,41.23) .. controls (115.24,43.99) and (113,46.23) .. (110.24,46.23) .. controls (107.48,46.23) and (105.24,43.99) .. (105.24,41.23) -- cycle ;
%Shape: Circle [id:dp6802297612715742]
\draw  [fill={rgb, 255:red, 0; green, 0; blue, 0 }  ,fill opacity=1 ] (105.24,61.23) .. controls (105.24,58.47) and (107.48,56.23) .. (110.24,56.23) .. controls (113,56.23) and (115.24,58.47) .. (115.24,61.23) .. controls (115.24,63.99) and (113,66.23) .. (110.24,66.23) .. controls (107.48,66.23) and (105.24,63.99) .. (105.24,61.23) -- cycle ;
%Shape: Circle [id:dp7471712696183264]
\draw  [fill={rgb, 255:red, 0; green, 0; blue, 0 }  ,fill opacity=1 ] (105.24,81.23) .. controls (105.24,78.47) and (107.48,76.23) .. (110.24,76.23) .. controls (113,76.23) and (115.24,78.47) .. (115.24,81.23) .. controls (115.24,83.99) and (113,86.23) .. (110.24,86.23) .. controls (107.48,86.23) and (105.24,83.99) .. (105.24,81.23) -- cycle ;
%Shape: Circle [id:dp054318459463163604]
\draw  [fill={rgb, 255:red, 0; green, 0; blue, 0 }  ,fill opacity=1 ] (125.24,101.23) .. controls (125.24,98.47) and (127.48,96.23) .. (130.24,96.23) .. controls (133,96.23) and (135.24,98.47) .. (135.24,101.23) .. controls (135.24,103.99) and (133,106.23) .. (130.24,106.23) .. controls (127.48,106.23) and (125.24,103.99) .. (125.24,101.23) -- cycle ;
%Shape: Circle [id:dp4582162391006499]
\draw  [fill={rgb, 255:red, 0; green, 0; blue, 0 }  ,fill opacity=1 ] (125.24,81.23) .. controls (125.24,78.47) and (127.48,76.23) .. (130.24,76.23) .. controls (133,76.23) and (135.24,78.47) .. (135.24,81.23) .. controls (135.24,83.99) and (133,86.23) .. (130.24,86.23) .. controls (127.48,86.23) and (125.24,83.99) .. (125.24,81.23) -- cycle ;
%Shape: Circle [id:dp6891859586617983]
\draw  [fill={rgb, 255:red, 0; green, 0; blue, 0 }  ,fill opacity=1 ] (125.24,60.23) .. controls (125.24,57.47) and (127.48,55.23) .. (130.24,55.23) .. controls (133,55.23) and (135.24,57.47) .. (135.24,60.23) .. controls (135.24,62.99) and (133,65.23) .. (130.24,65.23) .. controls (127.48,65.23) and (125.24,62.99) .. (125.24,60.23) -- cycle ;
%Shape: Circle [id:dp45399627548496446]
\draw  [fill={rgb, 255:red, 0; green, 0; blue, 0 }  ,fill opacity=1 ] (125.24,41.23) .. controls (125.24,38.47) and (127.48,36.23) .. (130.24,36.23) .. controls (133,36.23) and (135.24,38.47) .. (135.24,41.23) .. controls (135.24,43.99) and (133,46.23) .. (130.24,46.23) .. controls (127.48,46.23) and (125.24,43.99) .. (125.24,41.23) -- cycle ;
%Shape: Circle [id:dp36416029345788914]
\draw  [fill={rgb, 255:red, 0; green, 0; blue, 0 }  ,fill opacity=1 ] (125.24,22.23) .. controls (125.24,19.47) and (127.48,17.23) .. (130.24,17.23) .. controls (133,17.23) and (135.24,19.47) .. (135.24,22.23) .. controls (135.24,24.99) and (133,27.23) .. (130.24,27.23) .. controls (127.48,27.23) and (125.24,24.99) .. (125.24,22.23) -- cycle ;
%Shape: Circle [id:dp7385858689944478]
\draw  [fill={rgb, 255:red, 0; green, 0; blue, 0 }  ,fill opacity=1 ] (145.24,40.23) .. controls (145.24,37.47) and (147.48,35.23) .. (150.24,35.23) .. controls (153,35.23) and (155.24,37.47) .. (155.24,40.23) .. controls (155.24,42.99) and (153,45.23) .. (150.24,45.23) .. controls (147.48,45.23) and (145.24,42.99) .. (145.24,40.23) -- cycle ;
%Shape: Circle [id:dp39730901863962353]
\draw  [fill={rgb, 255:red, 0; green, 0; blue, 0 }  ,fill opacity=1 ] (165.24,60.23) .. controls (165.24,57.47) and (167.48,55.23) .. (170.24,55.23) .. controls (173,55.23) and (175.24,57.47) .. (175.24,60.23) .. controls (175.24,62.99) and (173,65.23) .. (170.24,65.23) .. controls (167.48,65.23) and (165.24,62.99) .. (165.24,60.23) -- cycle ;
%Shape: Circle [id:dp23939366838941223]
\draw  [fill={rgb, 255:red, 0; green, 0; blue, 0 }  ,fill opacity=1 ] (185.24,80.23) .. controls (185.24,77.47) and (187.48,75.23) .. (190.24,75.23) .. controls (193,75.23) and (195.24,77.47) .. (195.24,80.23) .. controls (195.24,82.99) and (193,85.23) .. (190.24,85.23) .. controls (187.48,85.23) and (185.24,82.99) .. (185.24,80.23) -- cycle ;
%Shape: Circle [id:dp7733682984906265]
\draw  [fill={rgb, 255:red, 0; green, 0; blue, 0 }  ,fill opacity=1 ] (145.24,80.23) .. controls (145.24,77.47) and (147.48,75.23) .. (150.24,75.23) .. controls (153,75.23) and (155.24,77.47) .. (155.24,80.23) .. controls (155.24,82.99) and (153,85.23) .. (150.24,85.23) .. controls (147.48,85.23) and (145.24,82.99) .. (145.24,80.23) -- cycle ;
%Shape: Circle [id:dp992208988564568]
\draw  [fill={rgb, 255:red, 0; green, 0; blue, 0 }  ,fill opacity=1 ] (165.24,21.23) .. controls (165.24,18.47) and (167.48,16.23) .. (170.24,16.23) .. controls (173,16.23) and (175.24,18.47) .. (175.24,21.23) .. controls (175.24,23.99) and (173,26.23) .. (170.24,26.23) .. controls (167.48,26.23) and (165.24,23.99) .. (165.24,21.23) -- cycle ;
%Shape: Circle [id:dp6672060820925345]
\draw  [fill={rgb, 255:red, 0; green, 0; blue, 0 }  ,fill opacity=1 ] (165.24,100.23) .. controls (165.24,97.47) and (167.48,95.23) .. (170.24,95.23) .. controls (173,95.23) and (175.24,97.47) .. (175.24,100.23) .. controls (175.24,102.99) and (173,105.23) .. (170.24,105.23) .. controls (167.48,105.23) and (165.24,102.99) .. (165.24,100.23) -- cycle ;
%Shape: Circle [id:dp027812763701182686]
\draw  [fill={rgb, 255:red, 0; green, 0; blue, 0 }  ,fill opacity=1 ] (145.24,101.23) .. controls (145.24,98.47) and (147.48,96.23) .. (150.24,96.23) .. controls (153,96.23) and (155.24,98.47) .. (155.24,101.23) .. controls (155.24,103.99) and (153,106.23) .. (150.24,106.23) .. controls (147.48,106.23) and (145.24,103.99) .. (145.24,101.23) -- cycle ;
%Shape: Circle [id:dp06997370981347983]
\draw  [fill={rgb, 255:red, 0; green, 0; blue, 0 }  ,fill opacity=1 ] (146,60.47) .. controls (146,57.71) and (148.24,55.47) .. (151,55.47) .. controls (153.76,55.47) and (156,57.71) .. (156,60.47) .. controls (156,63.23) and (153.76,65.47) .. (151,65.47) .. controls (148.24,65.47) and (146,63.23) .. (146,60.47) -- cycle ;
%Shape: Circle [id:dp3832734198293333]
\draw  [fill={rgb, 255:red, 0; green, 0; blue, 0 }  ,fill opacity=1 ] (165.24,81.23) .. controls (165.24,78.47) and (167.48,76.23) .. (170.24,76.23) .. controls (173,76.23) and (175.24,78.47) .. (175.24,81.23) .. controls (175.24,83.99) and (173,86.23) .. (170.24,86.23) .. controls (167.48,86.23) and (165.24,83.99) .. (165.24,81.23) -- cycle ;
%Shape: Circle [id:dp6798665118219767]
\draw  [fill={rgb, 255:red, 0; green, 0; blue, 0 }  ,fill opacity=1 ] (165.24,41.23) .. controls (165.24,38.47) and (167.48,36.23) .. (170.24,36.23) .. controls (173,36.23) and (175.24,38.47) .. (175.24,41.23) .. controls (175.24,43.99) and (173,46.23) .. (170.24,46.23) .. controls (167.48,46.23) and (165.24,43.99) .. (165.24,41.23) -- cycle ;
%Shape: Circle [id:dp8357313157791306]
\draw  [fill={rgb, 255:red, 0; green, 0; blue, 0 }  ,fill opacity=1 ] (165.24,41.23) .. controls (165.24,38.47) and (167.48,36.23) .. (170.24,36.23) .. controls (173,36.23) and (175.24,38.47) .. (175.24,41.23) .. controls (175.24,43.99) and (173,46.23) .. (170.24,46.23) .. controls (167.48,46.23) and (165.24,43.99) .. (165.24,41.23) -- cycle ;
%Shape: Circle [id:dp11763250065666964]
\draw  [fill={rgb, 255:red, 0; green, 0; blue, 0 }  ,fill opacity=1 ] (185.24,61.23) .. controls (185.24,58.47) and (187.48,56.23) .. (190.24,56.23) .. controls (193,56.23) and (195.24,58.47) .. (195.24,61.23) .. controls (195.24,63.99) and (193,66.23) .. (190.24,66.23) .. controls (187.48,66.23) and (185.24,63.99) .. (185.24,61.23) -- cycle ;
%Shape: Circle [id:dp13055485931192656]
\draw  [fill={rgb, 255:red, 0; green, 0; blue, 0 }  ,fill opacity=1 ] (185.24,41.23) .. controls (185.24,38.47) and (187.48,36.23) .. (190.24,36.23) .. controls (193,36.23) and (195.24,38.47) .. (195.24,41.23) .. controls (195.24,43.99) and (193,46.23) .. (190.24,46.23) .. controls (187.48,46.23) and (185.24,43.99) .. (185.24,41.23) -- cycle ;
%Shape: Circle [id:dp2953550763947854]
\draw  [fill={rgb, 255:red, 0; green, 0; blue, 0 }  ,fill opacity=1 ] (255,60.47) .. controls (255,57.71) and (257.24,55.47) .. (260,55.47) .. controls (262.76,55.47) and (265,57.71) .. (265,60.47) .. controls (265,63.23) and (262.76,65.47) .. (260,65.47) .. controls (257.24,65.47) and (255,63.23) .. (255,60.47) -- cycle ;
%Shape: Circle [id:dp7527807407345835]
\draw  [fill={rgb, 255:red, 0; green, 0; blue, 0 }  ,fill opacity=1 ] (256,19.47) .. controls (256,16.71) and (258.24,14.47) .. (261,14.47) .. controls (263.76,14.47) and (266,16.71) .. (266,19.47) .. controls (266,22.23) and (263.76,24.47) .. (261,24.47) .. controls (258.24,24.47) and (256,22.23) .. (256,19.47) -- cycle ;
%Shape: Circle [id:dp03179539931333397]
\draw  [fill={rgb, 255:red, 0; green, 0; blue, 0 }  ,fill opacity=1 ] (295,59.47) .. controls (295,56.71) and (297.24,54.47) .. (300,54.47) .. controls (302.76,54.47) and (305,56.71) .. (305,59.47) .. controls (305,62.23) and (302.76,64.47) .. (300,64.47) .. controls (297.24,64.47) and (295,62.23) .. (295,59.47) -- cycle ;
%Shape: Circle [id:dp5394700499366953]
\draw  [fill={rgb, 255:red, 0; green, 0; blue, 0 }  ,fill opacity=1 ] (336,99.47) .. controls (336,96.71) and (338.24,94.47) .. (341,94.47) .. controls (343.76,94.47) and (346,96.71) .. (346,99.47) .. controls (346,102.23) and (343.76,104.47) .. (341,104.47) .. controls (338.24,104.47) and (336,102.23) .. (336,99.47) -- cycle ;
%Shape: Circle [id:dp6252529131977465]
\draw  [fill={rgb, 255:red, 0; green, 0; blue, 0 }  ,fill opacity=1 ] (335,19.47) .. controls (335,16.71) and (337.24,14.47) .. (340,14.47) .. controls (342.76,14.47) and (345,16.71) .. (345,19.47) .. controls (345,22.23) and (342.76,24.47) .. (340,24.47) .. controls (337.24,24.47) and (335,22.23) .. (335,19.47) -- cycle ;
%Shape: Circle [id:dp9811903348351478]
\draw  [fill={rgb, 255:red, 0; green, 0; blue, 0 }  ,fill opacity=1 ] (336,60.47) .. controls (336,57.71) and (338.24,55.47) .. (341,55.47) .. controls (343.76,55.47) and (346,57.71) .. (346,60.47) .. controls (346,63.23) and (343.76,65.47) .. (341,65.47) .. controls (338.24,65.47) and (336,63.23) .. (336,60.47) -- cycle ;
%Shape: Circle [id:dp06402728438545191]
\draw  [fill={rgb, 255:red, 0; green, 0; blue, 0 }  ,fill opacity=1 ] (295,20.47) .. controls (295,17.71) and (297.24,15.47) .. (300,15.47) .. controls (302.76,15.47) and (305,17.71) .. (305,20.47) .. controls (305,23.23) and (302.76,25.47) .. (300,25.47) .. controls (297.24,25.47) and (295,23.23) .. (295,20.47) -- cycle ;
%Shape: Circle [id:dp09848166339871045]
\draw  [fill={rgb, 255:red, 0; green, 0; blue, 0 }  ,fill opacity=1 ] (295,101.47) .. controls (295,98.71) and (297.24,96.47) .. (300,96.47) .. controls (302.76,96.47) and (305,98.71) .. (305,101.47) .. controls (305,104.23) and (302.76,106.47) .. (300,106.47) .. controls (297.24,106.47) and (295,104.23) .. (295,101.47) -- cycle ;
%Shape: Circle [id:dp7160171375178561]
\draw  [fill={rgb, 255:red, 0; green, 0; blue, 0 }  ,fill opacity=1 ] (255,100.47) .. controls (255,97.71) and (257.24,95.47) .. (260,95.47) .. controls (262.76,95.47) and (265,97.71) .. (265,100.47) .. controls (265,103.23) and (262.76,105.47) .. (260,105.47) .. controls (257.24,105.47) and (255,103.23) .. (255,100.47) -- cycle ;
%Shape: Grid [id:dp9173097677304121]
\draw  [draw opacity=0] (380,39.48) -- (400.5,39.48) -- (400.5,60.48) -- (380,60.48) -- cycle ; \draw   (380,39.48) -- (380,60.48)(400,39.48) -- (400,60.48) ; \draw   (380,39.48) -- (400.5,39.48)(380,59.48) -- (400.5,59.48) ; \draw    ;
%Shape: Circle [id:dp9557844548541559]
\draw  [fill={rgb, 255:red, 0; green, 0; blue, 0 }  ,fill opacity=1 ] (385,49.47) .. controls (385,46.71) and (387.24,44.47) .. (390,44.47) .. controls (392.76,44.47) and (395,46.71) .. (395,49.47) .. controls (395,52.23) and (392.76,54.47) .. (390,54.47) .. controls (387.24,54.47) and (385,52.23) .. (385,49.47) -- cycle ;
%Shape: Grid [id:dp6541433148680511]
\draw  [draw opacity=0] (380,9.48) -- (400.5,9.48) -- (400.5,30.48) -- (380,30.48) -- cycle ; \draw   (380,9.48) -- (380,30.48)(400,9.48) -- (400,30.48) ; \draw   (380,9.48) -- (400.5,9.48)(380,29.48) -- (400.5,29.48) ; \draw    ;
%Shape: Circle [id:dp016626737428153904]
\draw   (380.5,19.47) .. controls (380.5,14.22) and (384.75,9.97) .. (390,9.97) .. controls (395.25,9.97) and (399.5,14.22) .. (399.5,19.47) .. controls (399.5,24.72) and (395.25,28.97) .. (390,28.97) .. controls (384.75,28.97) and (380.5,24.72) .. (380.5,19.47) -- cycle ;
%Shape: Grid [id:dp3831054493442876]
\draw  [draw opacity=0] (380,70.48) -- (400.5,70.48) -- (400.5,91.48) -- (380,91.48) -- cycle ; \draw   (380,70.48) -- (380,91.48)(400,70.48) -- (400,91.48) ; \draw   (380,70.48) -- (400.5,70.48)(380,90.48) -- (400.5,90.48) ; \draw    ;

% Text Node
\draw (42,83.48) node  [align=left] {$K$};
% Text Node
\draw (150,121.48) node  [align=left] {$B$};
% Text Node
\draw (302,119.48) node  [align=left] {$S$};
% Text Node
\draw (430,21.48) node  [align=left] {:Origin};
% Text Node
\draw (460,51.48) node  [align=left] {:Positions in set };
% Text Node
\draw (471, 81.48) node  [align=left] {:Positions not in set };

\end{tikzpicture}

\caption{The sets $K$, $S$ and $B$ used as an example; $K$ is the underlying
  structuring element, and $S$ is the sampling sieve.\label{fig:KBS} } 
\end{figure*}
%%%%%%%%%%%%%%%%%%%%%%%%%%%%%%%%%%%%%%%%%%%%%%%%%%%%%%%%%%%%%%

We observe that the sampling sieve $S$ as in the figure will return every second grid
point after sampling. Fixing the centre point of the structuring element $K$ at the same pixel
as the centre of the sampling sieve, we see that the range of the structuring element is smaller
than the distance between grid points of the sampling sieve. This amounts for a correct sampling
and can be used systematically for image reconstructions as documented by an example given
in Figure \ref{fig:F1}.

%%%%%%%%%%%%%%%%%%%%%%%%%%%%%%%%%%%%%%%%%%%%%%%%%%%%%%%%%%%%%%%%%%%%%%%%%%%%%%%%
%Figure F1
\begin{figure*}[!htbp]
	\includegraphics[width=0.3\linewidth]{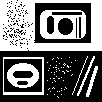}
\hfill
        \includegraphics[width=0.3\linewidth]{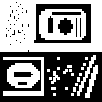}
 \hfill
 	\includegraphics[width=0.3\linewidth]{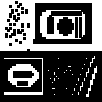}
        \vspace{1.0ex}
	\caption{Binary example image, of size $102\times 102$, $F_1$ (left), its maximal reconstruction after sampling, $(F_1\cap S)\oplus K$,  (centre) and its minimal reconstruction after sampling, $(F_1\cap S)\bullet K$, (right). For computing the maximal and minimal reconstruction, respectively, 
          the original image $F_1$ has first been sampled by $F_1\cap S$, which may reduce effectively the
          resolution (to a quarter of the original one)
          as only every second pixel is taken into account in both grid directions,
          see $S$ in Figure \ref{fig:KBS}.
          Then, by dilating with $K$ on the original grid, i.e.\ with the original resolution, we obtain
          as by the process of dilation an upper bound version of the original image called maximal reconstruction. The minimal reconstruction is obtained by closing the sample with $K$ on the original grid.
          }
	\label{fig:F1}
\end{figure*}
%%%%%%%%%%%%%%%%%%%%%%%%%%%%%%%%%%%%%%%%%%%%%%%%%%%%%%%%%%%%%%%%%%%%%%%%%%%%%%%%%%%

We now recall the \emph{binary version} of the \textit{Digital Morphological Sampling Theorem} from \cite{r2}.

\begin{theorem}{(Binary Digital Morphological Sampling Theorem)}
\label{thm:b1}
Let $F,\: K,\: S \in E^N$, where $F$ is the binary image, $S$ is the sampling sieve and $K$ is the structuring element used for filtering. Suppose $S$ and $K$ satisfy the sampling conditions \begin{enumerate}[I.]
\item \label{cond:1} $S\oplus S =S$
\item \label{cond:2} $S=\breve{S}$
\item	\label{cond:3} $K\cap S= \{0 \} $
\item \label{cond:4}	$K= \breve {K}$
\item \label{cond:5} $a\in K_b \Rightarrow K_a \cap K_b \cap S \neq \emptyset$
\end{enumerate}

Then,

\begin{enumerate}[I.]
\item $F\cap S\: =\: [(F\cap S)\bullet K]\cap S$
\item $F\cap S\: =\: [(F\cap S)\oplus K]\cap S$
\item $(F\cap S)\bullet K\: \subseteq  F\bullet K$
\item $(F\cap S)\oplus K\: \supseteq  F\circ K$
\item \label{res:5} If $F=F\circ K= F\bullet K$, then $(F\cap S) \bullet K\: \subseteq F\: \subseteq (F\cap S)\oplus K$
\item If $A=A\circ K$ and $F\cap S= A\cap S$, then $A\supseteq \: (F\cap S)\oplus K \: \Rightarrow A = \: (F\cap S)\oplus K$ 
\item If $A=A\bullet K$ and $F\cap S= A\cap S$, then $A\subseteq \: (F\cap S)\bullet K \: \Rightarrow A = \: (F\cap S)\bullet K$ 
\end{enumerate}

\end{theorem}

The conditions \ref{cond:4} and \ref{cond:5} imply that $S\oplus K = E^N$. The condition \ref{cond:3} implies that $K$ is just smaller than two sampling intervals. The Morphological Sampling Theorem states how the image must be filtered (i.e. opened or closed by $K$) to preserve the relevant information after sampling and gives set bounding relationships on reconstruction of morphologically filtered images.

We now proceed by reproducing some results given in \cite{r3} on relation between sampling the binary image after performing morphological operations and morphologically operating in the sampled domain. These results will still be useful later in the grey value setting, as these are concerned with the underlying set on which the operations are performed.
 
Let us note that the set $B$ is the SE which is used to perform morphological operations on the image $F_1$. The example figures in this section demonstrate the relationship between morphologically operating on the image $F_1$ with $B$, sampling using sieve $S$ and filter $K$.

Some of the following results, e.g.\ Theorem \ref{thm:b4} or Theorem \ref{thm:b2}, require that $B=B\circ K$. This in essence means that $B$ does not have any details (\textit{white region}) finer than the filter $K$, and $B$ can be appropriately reconstructed from the sampled SE, $B\cap S$, as mentioned in the  Binary Digital Morphological Sampling Theorem \ref{thm:b1}, cf.~result \ref{res:5}.

\begin{prop}
\label{thm:pb14}
Let $B\subseteq E^N$ be the structuring element. Then
\begin{enumerate}[I.]
\item $(F\cap S)\oplus (B\cap S)\: \subseteq \: (F\oplus B)\cap S$
\item $(F\cap S)\ominus (B\cap S)\: \supseteq \: (F\ominus B)\cap S$
\end{enumerate}
\end{prop}

Figures \ref{fig:bin1412} and \ref{fig:bin1411} illustrate the first part of above proposition. Figures \ref{fig:bin1422} and \ref{fig:bin1421} illustrate the second part of the proposition.

\begin{figure*}[!htbp]
\minipage{0.24\linewidth}
  \includegraphics[width=\linewidth]{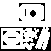}
  \caption{$(F_1\oplus B)\cap S$ \textcolor{white}{ABC} \textcolor{white}{ABCDE}}\label{fig:bin1412}
\endminipage\hfill
\minipage{0.24\linewidth}
  \includegraphics[width=\linewidth]{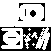}
  \caption{$(F_1\cap S)\oplus (B\cap S)$ \textcolor{white}{ABCDE}}\label{fig:bin1411}
\endminipage\hfill
\minipage{0.24\linewidth}
  \includegraphics[width=\linewidth]{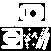}
  \caption{$ \{ [(F_1\cap S)\bullet K]\oplus B\} \cap S$}\label{fig:bin22}
\endminipage\hfill
\end{figure*}

\begin{figure*}[!htbp]
\minipage{0.24\linewidth}
  \includegraphics[width=\linewidth]{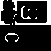}
  \caption{$(F_1\ominus B)\cap S$ \textcolor{white}{ABC} \textcolor{white}{ABCDE}}\label{fig:bin1422}
\endminipage\hfill
\minipage{0.24\linewidth}
  \includegraphics[width=\linewidth]{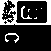}
  \caption{$(F_1\cap S)\ominus (B\cap S)$ \textcolor{white}{ABCDE}}\label{fig:bin1421}
\endminipage\hfill
\minipage{0.24\linewidth}
  \includegraphics[width=\linewidth]{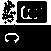}
  \caption{$ \{ [(F_1\cap S)\oplus K]\ominus B\} \cap S$}\label{fig:bin32}
\endminipage\hfill
\end{figure*}

\begin{lemma}
\label{thm:lema}
\begin{enumerate}[I.]
\item $(F\cap S)\oplus (B\cap S) \:=\: [F\oplus (B\cap S)]\cap S$
\item $(F\cap S)\ominus (B\cap S)\: =\: [F\ominus (B\cap S)]\cap S$
\end{enumerate}
\end{lemma}

\begin{lemma}
\label{thm:lemb}
Let $B=B\circ K$ . Then $[(F\cap S)\bullet K]\oplus B \subseteq \: [(F\cap S)\oplus K]\oplus (B\cap S)$
\end{lemma}

The following two theorems formulated in \cite{r2}, as indicated for binary images,
elaborate on interaction of sampling and dilation/erosion, respectively.  Let us note that these theorems serve as a motivation
for proving and validating corresponding theorems in the grey value setting later.

\begin{theorem}{Sample Dilation Theorem.}
\label{thm:b2}
Let $B=B\circ K$. Then $(F\cap S)\oplus (B\cap S) =\: \{ [(F\cap S)\bullet K]\oplus B\} \cap S $
\end{theorem}

\begin{theorem} {Sample Erosion Theorem.}
\label{thm:b3}
Let $B=B\circ K$ Then $(F\cap S)\ominus (B\cap S)=\: \{[(F\cap S)\oplus K]\ominus B\}\cap S$.
\end{theorem}

Notice that Figures \ref{fig:bin1412} and \ref{fig:bin22} are identical. These figures demonstrate Sample Dilation Theorem \ref{thm:b2}, meaning that dilation in sampled domain (here, $(F_1\cap S)\oplus (B\cap S)$) is equivalent to sampling after dilation of the minimal reconstruction (here, $\{ [(F_1\cap S)\bullet K]\oplus B\} \cap S$). 

Similarly, Figures \ref{fig:bin1421} and \ref{fig:bin32} are identical. This pair demonstrates Sample Erosion Theorem \ref{thm:b3}, meaning that erosion in sampled domain (here, $(F_1\cap S)\ominus (B\cap S)$) is equivalent to sampling after erosion of the maximal reconstruction (here, $\{ [(F_1\cap S)\oplus K]\oplus B\} \cap S$). 

\begin{prop}
\label{thm:pb16}
$[F\circ(B\cap S)]\cap S =\: (F\cap S)\circ (B\cap S)$
\end{prop}

\begin{prop}
\label{thm:pb17}
$[F\bullet (B\cap S)]\cap S =(F\cap S)\bullet (B\cap S)$
\end{prop}

Similarly to Theorems \ref{thm:b2}, \ref{thm:b3}, the following two theorems provide the interaction of two other fundamental morphological operations, closing and opening, with sampling. We will also extend the following result to the grey value setting. 

\begin{theorem}{Sample Opening and Closing Bounds Theorem.}
\label{thm:b4}
Suppose $B=B\circ K$, then 
\begin{enumerate}[I.]
\item $ \{ F\circ [(B\cap S)\oplus K]\} \cap S \subseteq \: (F\cap S)\circ (B\cap S) \subseteq \: \{[(F\cap S)\oplus K]\circ B\} \cap S$
\item $\{ [(F \cap S)\bullet K]\bullet B\} \cap S \subseteq \: (F\cap S)\bullet (B\cap S) \subseteq \: \{F\bullet [(B\cap S)\oplus K] \} \cap S$
\end{enumerate}
\end{theorem}    

Figures \ref{fig:bin411}-\ref{fig:bin413} demonstrate interaction of sampling with opening operation. We can observe that opening in the sampled domain is bounded above by sampling the opening of maximal reconstruction of image (here, $\{[(F_1 \cap S)\oplus K]\circ B\} \cap S$) and bounded below by sampling of opening by maximal reconstruction of filter (here, $\{ F_1\circ [(B\cap S)\oplus K]\} \cap S $).  

In a similar way, Figures \ref{fig:bin421}-\ref{fig:bin423} demonstrate interaction of sampling with closing operation. We see that closing in sampled domain is bounded by sampling of closing minimal reconstruction of the image (here, $\{ [(F_1 \cap S)\bullet K]\bullet B\} \cap S $) and sampling of closing by maximal reconstruction of the filter (here, $\{F_1 \bullet [(B\cap S)\oplus K] \} \cap S $).

\begin{figure*}[!htbp]
\minipage{0.24\linewidth}
  \includegraphics[width=\linewidth]{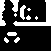}
  \caption{$\{ F_1\circ [(B\cap S)\oplus K]\} \cap S $ \textcolor{white}{ABCDE}}\label{fig:bin411}
\endminipage\hfill
\minipage{0.24\linewidth}
  \includegraphics[width=\linewidth]{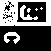}
  \caption{$ (F_1\cap S)\circ (B\cap S)$ \textcolor{white}{ABCDE}}\label{fig:bin412}
\endminipage\hfill
\minipage{0.24\linewidth}
  \includegraphics[width=\linewidth]{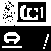}
  \caption{$ \{[(F_1 \cap S)\oplus K]\circ B\} \cap S$}\label{fig:bin413}
\endminipage\hfill
\end{figure*}

\begin{figure*}[!htbp]
\minipage{0.24\linewidth}
  \includegraphics[width=\linewidth]{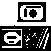}
  \caption{$\{ [(F_1 \cap S)\bullet K]\bullet B\} \cap S$}\label{fig:bin421}
\endminipage\hfill
\minipage{0.24\linewidth}
  \includegraphics[width=\linewidth]{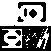}
  \caption{$(F_1 \cap S)\bullet (B\cap S)$ \textcolor{white}{ABCDE}}\label{fig:bin422}
\endminipage\hfill
\minipage{0.24\linewidth}
  \includegraphics[width=\linewidth]{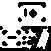}
  \caption{$\{F_1 \bullet [(B\cap S)\oplus K] \} \cap S$}\label{fig:bin423}
\endminipage\hfill
\end{figure*}

\begin{theorem} {Sampling Opening And Closing Theorem.}
\label{thm:b5}
Suppose $B=B\circ K$. \begin{enumerate}[I.]
\item If $F=(F\cap S)\oplus K$ and $B=(B\cap S)\oplus K$, then $(F\cap S)\circ (B\cap S) =\: (F\circ B)\cap S$
\item If $F=(F\cap S)\bullet K$ and $B=(B\cap S)\oplus K$, then $(F\cap S)\bullet (B\cap S) =\: (F\bullet B)\cap S$
\end{enumerate}
\end{theorem}

%%%%%%%%%%%%%%%%%%%%%%%%%%%%%%%%%%%%%%%%%%%%%%%%%%%%%%%%
\subsection{Morphological Notions for Grey-value Images} 

Let $E$ be the set of integers used for denoting the indices of the coordinates. 
A grey-value image is represented by a function $f: F \rightarrow L$,
$F \subseteq E^{N}$, $L=[0,l]$, where $l > 0$ is the upper limit for grey values 
at a pixel in grey-value image, and $N=2$ for two dimensional grey-value images.  

The SEs are of finite size. The morphological operations require taking 
$\max$ or $\min$ of grey values over finite sets (of pixels) in our setting. 
Therefore, the results are independent on whether $L$ is a discrete set 
(e.g., subset of integers) or a continuous set (e.g., sub-interval of real numbers).

The proposed extension of the previous notions to grey-value images requires to 
define the notions of top surfaces and the umbra of an image, compare \cite{r1}. 
Also see Figure \ref{fig:umbra} for the latter. In a first step we rely on both 
notions for defining dilation and erosion.

\begin{definition}{\textbf{Top Surface.}}
 \label{def:topsur}
 %changes done here on 10th September 2019
 Let $A\subseteq E^{N} \times L$. and $F=\{x\in E^{N} \vert\: \text{for some }y\in L,\: (x,y)\in A\}$. Then the top surface of $A$ is denoted as $T[A]$ and defined as $T[A](x)\: =\: \max \{y\vert\:(x,y)\in A \}$.   
 \end{definition}

\begin{definition}{\textbf{Umbra of an Image.}}
\label{def:umbraImg}
Let $F\subseteq E^{N}$ and $f:F\rightarrow L$. The umbra of the image $f$ is denoted by $U[f]$ and is defined as $U[f]=\: \{(x,y)\vert\: x\in F \text{, } y\in L \text{ and } 0\leq y\leq f(x) \}$.
\end{definition}

 \begin{figure*}[!htbp]
  \includegraphics[width=0.3\linewidth]{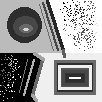}
  \hfill
  \includegraphics[width=0.6\linewidth]{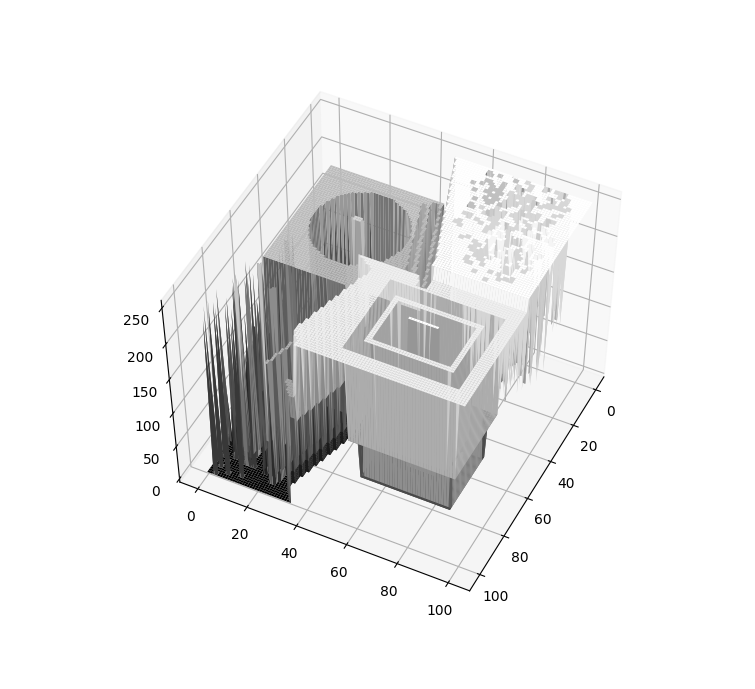}

  \vspace{1.0ex}
  \caption{\textbf{Left:} An example grey-value image, of size $102\times 102$, $f:F\rightarrow L$ used for examples. \textbf{Right:} Visualisation of the corresponding umbra.\label{fig:umbra}}
\end{figure*}

%Though lattice theory provides a rich framework (refer to \cite{r4, r5}) for dealing with morphological operations, but it lacks the tools to deal with non-flat SEs. 
The umbra is useful for studying geometric relations between pixels, as it makes use of both 
spatial and tonal information. Thus, umbra approach provides important and convenient 
tools to deal with non-flat SEs and sampling, as already briefly discussed 
in the introduction. 
%\cite{r1, r2}. 
Using umbra allows us to describe grey-value morphological operations in terms of 
corresponding binary operations of umbra, making use of SE on umbra of the image.
Thereby, the SE by itself is also transferred to the umbra setting.
See Definition \ref{def:gdilation} for dilation, Definition \ref{def:gerosion} for 
erosion and Proposition \ref{prop:AltGrayClosing} for opening and closing. In this way
we are able to extend many morphological concepts developed for binary images directly 
to grey-value images.   

Our definition of umbra, compare Figure \ref{fig:umbra}, turns out to work perfectly well
for grey-value images, compare \cite{r1}. Note that we restrict ourselves here to 
discrete non-negative pixel values at discrete positions (pixels). The underlying concept cannot 
be directly extended to continuous domain or negative values, compare \cite{r6,r7}.

\begin{definition}{\textbf{Dilation for Grey Value Images.}}
\label{def:gdilation}
Let $F,\: K \subseteq E^{N} $ and $f:F\rightarrow L$, $k:K\rightarrow L$. The dilation of $f$ by $k$ is denoted by $f\oplus k:\: F\oplus K\rightarrow L $ and is defined as $f\oplus k=\: T[U[f] \oplus U[k]]$.
\end{definition}

\begin{definition}{\textbf{Erosion for Grey Value Images.}}
\label{def:gerosion}
Let $F,\: K \subseteq E^{N} $ and $f:F\rightarrow L$, $k:K\rightarrow L$. The erosion of $f$ by $k$ is denoted by $f\ominus k:\: F\ominus K\rightarrow L $ and is defined as $f\ominus k=\: T[U[f] \ominus U[k]]$.
\end{definition}

As already mentioned, the operations of opening and closing can easily be extended from binary images 
to the grey-value setting, following the same combination of dilation/erosion as in Definition \ref{binary-opening-closing}. 
That is, $f\circ k$ $=$ $(f\ominus k)\oplus k$ and $f\bullet k$ $=$ $(f\oplus k)\ominus k$.
Similarly to the binary versions, grey-value opening and closing are anti-extensive and extensive, 
respectively. Both grey-value opening and closing are idempotent as well as dual operations.

Propositions \ref{thm:p50} and  \ref{thm:prop52} are used to compute dilation and erosion 
for a grey-value image. %For more details compare \cite{HSZ-87}.

\begin{prop}
\label{thm:p50}
 Let $F,\: K \subseteq E^{N} $ and $f:F\rightarrow L$, $k:K\rightarrow L$.  
 Then $f\oplus k:\: F\oplus K\rightarrow L $ can be computed by $(f\oplus k)(x) =\: \max_{u\in K,\: x-u\in F}  {\{f(x-u)\: +\: k(u) \}}$.
\end{prop}

\begin{prop}
\label{thm:prop52}
Let $F,\: K \subseteq E^{N} $ and $f:F\rightarrow L$, $k:K\rightarrow L$.  
Then $f\ominus k:\: F\ominus K\rightarrow L $ can be computed 
by $(f\ominus k)(x) =\: \min_{u\in K}  {\{f(x+u)\: -\: k(u) \}} $.
\end{prop}

Let us note that the method to compute  $(f\ominus k)(x) $ as given in 
Proposition \ref{thm:prop52} 
is valid only for $x$ such that $(x,y) \in U[f]\ominus U[k]$ for some $y$.
To extend the definition to all $x\in F\ominus K$, one may define  
$(f\ominus k)(x) =\: \max \{0, \min_{u\in K}  {\{f(x+u)\: -\: k(u) \}} \}$.

Suppose now $A$ and $B$ are umbras. Then $A\oplus B$ and $A\ominus B$ are umbras. 
The Umbra Homomorphism Theorem below makes this property precise. 

\begin{theorem}{\textbf{Umbra Homomorphism Theorem.}}
\label{thm:58}
Let $F,\: K\subset E^{N}$ and $f:F\rightarrow L$ and $k:K\rightarrow L$. Then \begin{enumerate}[I.]
\item $U[f\oplus k]=U[f]\oplus U[k]$
\item	$U[f\ominus k]= U[f]\ominus U[k]$
\end{enumerate}
\end{theorem}

We will employ at some point the following notions.

\begin{definition}{\textbf{Reflection of an Image.}}
\label{def:RefImg}
The reflection of a grey-value image $f:F\rightarrow L$ is denoted by $\breve{f}:\breve{F} \rightarrow L$, and it
is defined as $\breve{f}(x)=f(-x)$ for each $x\in \breve{F}$.  
\end{definition}

\begin{definition}{\textbf{Negative of an Image.}}
\label{def:NegImg}
The negative of a grey-value image  $f:F\rightarrow L$ is denoted by $-f:F\rightarrow L$, and it is defined
as $(-f)(x) = l-f(x)$ for each $x\in F$, where $l > 0$ is again the upper limit for possible grey values.
\end{definition}

Let us note that Haralick \cite{r1} defines the negative of an image via $(-f)(x)=-f(x), x\in F$. Definition \ref{def:NegImg} 
as we propose above is more suitable for many purposes 
because the grey value at a pixel is not supposed to be negative.

We also recall the concept of \textit{boundedness} for grey-value 
images. Let $f:F\rightarrow L$ and $g:G\rightarrow L$ be two grey-value 
images. We say $f\leq g$ if $F\subseteq G$ and $f(x)\leq g(x)$ 
for each $x\in F$.

Let us now recite the grey-value morphological 
sampling theorem from \cite{r2}. 

\begin{theorem} {\textbf{The Grey Scale Digital Morphological Sampling Theorem.}}
\label{thm:b6}
Let $F,K,S \subseteq E^{N}$, $f:F\rightarrow L$ is the image,  $k:K\rightarrow L$ is the structuring element used for filtering. Let $K, S$, and $k:K\rightarrow L$ satisfy the following conditions:
\begin{enumerate}[I.]
\item $S\oplus S =S$
\item $S=\breve{S}$
\item $K\cap S =\{ 0 \}$
\item \label{cond:g4}$a\in K_b \Rightarrow K_a \cap K_b \cap S \neq \emptyset$
\item \label{cond:g5}$k=\breve{k}$
\item \label{cond:g6}$k(a) \leq k(a-b) +k(b) $, $\forall a, b\in K$ with $a-b \in K$ 
\item \label{cond:g7}$k(0)=0$
\end{enumerate}
Then, 
\begin{enumerate}[I.]
\item $f\vert_S =\: (f\vert_S \bullet k)\vert_S$
\item $f\vert_S =\: (f\vert_S \oplus k)\vert_S$
\item $f\vert_S \bullet k \leq \: f\bullet k$
\item $f \circ k \leq \: f\vert_S \oplus k$
\item \label{res:g5} If $f=f\circ k$ and $f=f\bullet k$, then  $f\vert_S \bullet k \leq \: f \leq \: f\vert_S \oplus k$
\item If $g=g\bullet k$ , $g\vert_S =f\vert_S$ and $g\leq f\vert_S \bullet k$ then $g=\: f\vert_S \bullet k$
\item If $g=g\circ k$ , $g\vert_S =f\vert_S$ and $g\geq f\vert_S \oplus k$ then $g=\: f\vert_S \oplus k$
\end{enumerate}
\end{theorem}

The conditions \ref{cond:g5},  \ref{cond:g6} and \ref{cond:g7} are introduced in \cite{r2} to allow for proper sampling and reconstruction. The conditions \ref{cond:g5} and \ref{cond:g6} imply $k(y)\geq 0 \: \, \forall y\in K$. The class of flat SEs symmetric about the origin as well as the class of paraboloid SEs with $k(0)=0$ and 
symmetric about the origin, satisfy these conditions.

We recall one more proposition from \cite{r2}, which is employed in proofs of some results in the next section. 

\begin{prop}
\label{thm:pb22}
Let $F,K,S \subseteq E^{N}$, $f:F\rightarrow L, \: k:K\rightarrow L$. Suppose $k=\breve{k}$, $u\in K_v$ imply $S\cap K_u \cap K_v \neq \emptyset$ and $k(a) \leq k(a-b) +k(b) $, $\forall a, b\in K$ satisfying $a-b \in K$. 

Then for every $u\in K$, $f(x+z)-k(z) \leq \: (f\vert_S \oplus k)(x+u)-k(u)$ for each $z\in K $ satisfying $u-z \in K$.
\end{prop}
 
%%%%%%%%%%%%%%%%%%%%%%%%%%%%%%%%%%%%%%%%%%%%%%%%%%%%%%%%%%%%%%%%%%%%%%%%%% 

\section{New Extensions of the Classic Morphological Sampling Theorem} %change name of the section?

Before discussing sampling in the forthcoming subsection, 
let us introduce the notion of the reflection of the umbra useful
in our setting. Furthermore, we prove an umbra-based 
monotonicity principle.

\begin{definition}{\textbf{Reflection of Umbra.}} 
Let $A\subseteq E^{N}\times L$ be a non-empty set (not necessarily an umbra). 
Then the reflection of $A$ is denoted by $\tilde{A}$ and is defined as 
$\tilde{A}$ $=$ $\{(x,a) \vert $ $(-x,y)\in A$ for 
some $y\in L,$  $l-T[A](-x)\: \leq a\leq l  \}$.  
\end{definition}

\begin{prop}{\textbf{Umbra Monotonicity Principle.}}
\label{thm:refUm}
Let $A\subseteq E^{N}\times L$  and $B\subseteq E^{N}\times L$ be two non-empty sets. If $A\subseteq B$, then $\tilde{A} \subseteq \tilde{B}$.  
\end{prop}

\begin{proof}
\begin{align*}
(x, a)\in  \tilde{A}  \Rightarrow \: &(-x,y_1)\in A  \text{ for some $y_1 \in L$}\\
& \text{ and } l-T[A](-x)\leq a \leq l\\
 \Rightarrow \: & (-x,y_2) \in B \text{ for some $y_2 \in L$ and} \\
   &l-T[B](-x)\leq l-T[A](-x)\leq a \leq l\\
 \Rightarrow \: &(x,a)\in  \tilde{B}
\end{align*}
\end{proof}

We will make use also of the following notion.

\begin{definition}{\textbf{Translation of an Umbra.}}
Let $A\subseteq E^{N}\times L$ be a non-empty set (not necessarily an umbra) and $y_0 \geq 0$. Then, $A_{(x_0, y_0)}$ $=$ $\{(x+x_0, y+y_0)\vert (x,y)\in A \text{ and } 0\leq y+y_0 \leq l \}$.
\end{definition}

In our work, we will employ the following alternative 
definitions of grey scale opening and closing.
Our notions rely on the alternative definitions of 
opening and closing for grey value images as formulated in
\eqref{alt-open} and \eqref{alt-close} for binary images.

\begin{prop} {\textbf{Alternative Definition of Grey Scale Opening / Closing.}}
\begin{align}
\begin{split}
U[f\circ k] = &\bigcup _{\{ (x,y) \vert  U[k]_{(x,y)} \subseteq U[f]\} } U[k]_{(x,y)}
\end{split}
\label{thm:p71}
\end{align}

\begin{align}
\begin{split}
U[f\bullet k] = \: & U[(-((-f)\circ \breve{k}))] \\
= \: & \{ (x,y) \vert    (x,y)\in  \tilde{U[k]_{(x_0,y_0)}}\\
&\: \mathrm{ implies } \; \tilde{U[k]_{(x_0,y_0)}} \cap U[f] \neq \emptyset \}
\end{split}
\label{prop:AltGrayClosing}
\end{align}
\end{prop}

\begin{proof} 
The alternative definition of opening directly follows Umbra Homomorphism Theorem and definition of opening for grey-value images.

 We elaborate here on the last equality in \eqref{prop:AltGrayClosing}.\\
$(x,y)  \in U[f\bullet k] = U[(-((-f)\circ k))]$\\
\\
 $\Leftrightarrow$  for any  $\alpha > 0$,  $(x,l+\alpha -y)$ $\not \in$   $U[((-f)\circ k)]$\\
 \\ 
 $\Leftrightarrow$  for any  $\alpha > 0$,  for any  $(x_0 ,y_0)$, $y_0 \geq 0$,\\
 satisfying  $(x,l+\alpha-y)$ $\in U[\breve{k}] _{(x_o,y_0+\alpha )}$,\\
 we have $U[\breve{k}] _{(x_o,y_0+\alpha )}$ $\not \subseteq U[(-f)]$\\
\\
$\Leftrightarrow$  for any  $\alpha > 0$,  for any  $(x_0 ,y_0)$,  $y_0 \geq 0$, \\
satisfying  $(x,l+\alpha-y)$ $\in$ $U[\breve{k}] _{(x_o,y_0+\alpha )}$,\\
$\exists$ $u$ $\in$ $\breve{K}:$ $T[U[\breve{k}] _{(x_o,y_0+\alpha )}] (u+x_0) > l-f(u+x_0)$\\
\\
$\Leftrightarrow$  for any  $\alpha > 0$,  for any  $(x_0 ,y_0)$, $y_0 \geq 0$,\\ 
satisfying $(x,l+\alpha-y)$ $\in$ $U[\breve{k}] _{(x_o,y_0+\alpha )}$, \\
 $\exists u$ $\in$ $\breve{K}:$  $f(u+x_0) >$  $l -T[U[\breve{k}] _{(x_0,y_0+\alpha )}](u+x_0)$\\
 \\
$\Leftrightarrow$ for any  $(x_0,y_0)$,  $y_0 \geq 0$, \\
 satisfying $(x,l-y)$ $\in$   $U[\breve{k}] _{(x_0,y_0)}$,\\ 
 $\exists u$ $\in$  $\breve{K}:$    $f(u+x_0) \geq$  $l -T[U[\breve{k}] _{(x_0,y_0)}](u+x_0)$\\
 \\
$\Leftrightarrow$ $(x,y)\in$ $\tilde{U[k]_{(-x_0,y_0)}}$  implies \\
$\tilde{U[k]_{(-x_0,y_0)}} \cap U[f]$ $\neq$ $\emptyset$ \\
\\i.e.\, for any $(x,y)\in E^N \times L$,      $(x,y)\in \tilde{U[k]_{(x_0,y_0)}}$  implies 
 $\tilde{U[k]_{(x_0,y_0)}} \cap U[f]$ $\neq$ $\emptyset$.

\end{proof}

In \cite{r1}, the authors describe the effect of closing with a paraboloid SE as 
"taking the reflection of the paraboloid, turning it upside down and sliding 
it all over the top surface of $f$. The closing is the surface of all the 
lowest point reached by the sliding paraboloid".
The Proposition \ref{prop:AltGrayClosing} demonstrates 
this mathematically, for all considered types of SEs.
%include citation to our conference paper here?

\subsection{Morphologically Operating in the Sampled Domain}
\label{Sec:OpSampDom}

In this section, we present the main results of this paper. 
We examine the relation between sampling after performing the 
morphological operations and morphological operating in sampled domains 
on grey-value images. 

We study morphologically operating on image $f:F\rightarrow L$ by 
a SE $b:B\rightarrow L$ with respect to sampling, using a sieve $S$. 
The SE $k:K\rightarrow L$ is used to filter the image $f$ as well as the SE $b$, 
for sampling. $k$ is also used for the purpose of reconstruction.    
We assume that $S, K \subseteq \: E^{N}$ and $k:K\rightarrow L$ 
satisfy the conditions mentioned in The Grey-Value Digital 
Morphological Sampling Theorem (\ref{thm:b6}).

As in the case of binary images, some of the following results, e.g.\ Theorem \ref{thm:2ba}, Theorem \ref{thm:b4a} etc.\, require that $b=b\circ k$. 
This in essence means that $b$ does not have any details finer than 
the filter $k$, and $b$ can be appropriately reconstructed from the 
sampled SE $b\vert _S$ as mentioned in the  Grey-value Digital Morphological Sampling Theorem \ref{thm:b6}, result \ref{res:g5}.

It is interesting to note that due to idempotence of morphological 
opening, given any SE $b:B\rightarrow L$, when filtered 
(using morphological opening) with $k$, 
gives $b\circ k=$ $b_{\text{filt}}:B\circ K\rightarrow L$, and 
$b_{\text{filt}}$ satisfies the property 
$b_{\text{filt}} \circ k$ $=(b\circ k)\circ k$ $=b\circ k$ $=b_{\text{filt}}$.

%%%%%%%%%%%%%%%%%%%%%%%%%%%%%%%%%%%%%%%%%%%%%%%%%%%%%%%%%%%%%%%%%%%%%%%%%%%%%%%%%%%%%%%%%%%%%%%%%%%%%%%%%%%%%%%%%%%%%%%%%%%%%%%%%%%
% Figure kb
\begin{figure*}[!htbp]
\centering

\tikzset{every picture/.style={line width=0.75pt}} %set default line width to 0.75pt       

\begin{tikzpicture}[x=0.65pt,y=0.65pt,yscale=-1,xscale=1]
%uncomment if require: \path (0,140.484375); %set diagram left start at 0, and has height of 140.484375

%Shape: Grid [id:dp9597332067298479]
\draw  [draw opacity=0] (11,10) -- (71.5,10) -- (71.5,70.48) -- (11,70.48) -- cycle ; \draw   (11,10) -- (11,70.48)(31,10) -- (31,70.48)(51,10) -- (51,70.48)(71,10) -- (71,70.48) ; \draw   (11,10) -- (71.5,10)(11,30) -- (71.5,30)(11,50) -- (71.5,50)(11,70) -- (71.5,70) ; \draw    ;
%Shape: Grid [id:dp1651748896103642]
\draw  [draw opacity=0] (101,10) -- (201.5,10) -- (201.5,110.48) -- (101,110.48) -- cycle ; \draw   (101,10) -- (101,110.48)(121,10) -- (121,110.48)(141,10) -- (141,110.48)(161,10) -- (161,110.48)(181,10) -- (181,110.48)(201,10) -- (201,110.48) ; \draw   (101,10) -- (201.5,10)(101,30) -- (201.5,30)(101,50) -- (201.5,50)(101,70) -- (201.5,70)(101,90) -- (201.5,90)(101,110) -- (201.5,110) ; \draw    ;
%Shape: Circle [id:dp7724242829636028]
\draw   (31.5,40.23) .. controls (31.5,34.85) and (35.86,30.48) .. (41.24,30.48) .. controls (46.62,30.48) and (50.98,34.85) .. (50.98,40.23) .. controls (50.98,45.61) and (46.62,49.97) .. (41.24,49.97) .. controls (35.86,49.97) and (31.5,45.61) .. (31.5,40.23) -- cycle ;
%Shape: Circle [id:dp9982013591868288]
\draw   (141.5,60.47) .. controls (141.5,55.22) and (145.75,50.97) .. (151,50.97) .. controls (156.25,50.97) and (160.5,55.22) .. (160.5,60.47) .. controls (160.5,65.72) and (156.25,69.97) .. (151,69.97) .. controls (145.75,69.97) and (141.5,65.72) .. (141.5,60.47) -- cycle ;
%Shape: Grid [id:dp9173097677304121]
\draw  [draw opacity=0] (273,40.48) -- (293.5,40.48) -- (293.5,61.48) -- (273,61.48) -- cycle ; \draw   (273,40.48) -- (273,61.48)(293,40.48) -- (293,61.48) ; \draw   (273,40.48) -- (293.5,40.48)(273,60.48) -- (293.5,60.48) ; \draw    ;
%Shape: Grid [id:dp6541433148680511]
\draw  [draw opacity=0] (273,10.48) -- (293.5,10.48) -- (293.5,31.48) -- (273,31.48) -- cycle ; \draw   (273,10.48) -- (273,31.48)(293,10.48) -- (293,31.48) ; \draw   (273,10.48) -- (293.5,10.48)(273,30.48) -- (293.5,30.48) ; \draw    ;
%Shape: Circle [id:dp016626737428153904]
\draw   (273.5,20.47) .. controls (273.5,15.22) and (277.75,10.97) .. (283,10.97) .. controls (288.25,10.97) and (292.5,15.22) .. (292.5,20.47) .. controls (292.5,25.72) and (288.25,29.97) .. (283,29.97) .. controls (277.75,29.97) and (273.5,25.72) .. (273.5,20.47) -- cycle ;
%Shape: Grid [id:dp3831054493442876]
\draw  [draw opacity=0] (273,71.48) -- (293.5,71.48) -- (293.5,92.48) -- (273,92.48) -- cycle ; \draw   (273,71.48) -- (273,92.48)(293,71.48) -- (293,92.48) ; \draw   (273,71.48) -- (293.5,71.48)(273,91.48) -- (293.5,91.48) ; \draw    ;

% Text Node
\draw (42,83.48) node  [align=left] {$\displaystyle k:K\rightarrow \mathbb{Z}$};
% Text Node
\draw (150,121.48) node  [align=left] {$\displaystyle b:B\rightarrow \mathbb{Z}$};
% Text Node
\draw (326,22.48) node  [align=left] {:Origin};
% Text Node
\draw (425,50.48) node  [align=left] {:Value of function is $x$ at the position};
% Text Node
\draw (393,81.48) node  [align=left] {:Position not in the domain};
% Text Node
\draw (283,49.48) node  [align=left] {$\displaystyle x$};
% Text Node
\draw (41,41.48) node  [align=left] {0};
% Text Node
\draw (61,60.48) node  [align=left] {0};
% Text Node
\draw (21,21.48) node  [align=left] {0};
% Text Node
\draw (41,21.48) node  [align=left] {0};
% Text Node
\draw (61,21.48) node  [align=left] {0};
% Text Node
\draw (21,41.48) node  [align=left] {0};
% Text Node
\draw (41,61.48) node  [align=left] {0};
% Text Node
\draw (21,60.48) node  [align=left] {0};
% Text Node
\draw (61,41.48) node  [align=left] {0};
% Text Node
\draw (111,41.48) node  [align=left] {0};
% Text Node
\draw (131,61.48) node  [align=left] {0};
% Text Node
\draw (151,81.48) node  [align=left] {0};
% Text Node
\draw (171,101.48) node  [align=left] {0};
% Text Node
\draw (111,61.48) node  [align=left] {0};
% Text Node
\draw (131,81.48) node  [align=left] {0};
% Text Node
\draw (151,101.48) node  [align=left] {0};
% Text Node
\draw (131,101.48) node  [align=left] {0};
% Text Node
\draw (171,81.48) node  [align=left] {0};
% Text Node
\draw (131,21.48) node  [align=left] {0};
% Text Node
\draw (151,41.48) node  [align=left] {0};
% Text Node
\draw (171,61.48) node  [align=left] {0};
% Text Node
\draw (191,81.48) node  [align=left] {0};
% Text Node
\draw (111,81.48) node  [align=left] {0};
% Text Node
\draw (131,41.48) node  [align=left] {0};
% Text Node
\draw (151,21.48) node  [align=left] {0};
% Text Node
\draw (171,41.48) node  [align=left] {0};
% Text Node
\draw (191,61.48) node  [align=left] {0};
% Text Node
\draw (171,21.48) node  [align=left] {0};
% Text Node
\draw (191,41.48) node  [align=left] {0};
% Text Node
\draw (151,61.48) node  [align=left] {0};

\end{tikzpicture}
\caption{The structuring elements $k$ and $b$ used in the examples}
\label{fig:kb-grey}
\end{figure*}
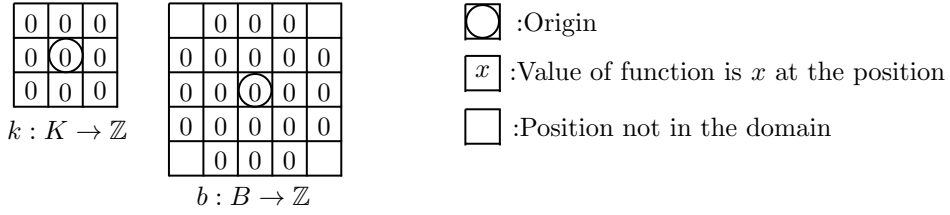
%%%%%%%%%%%%%%%%%%%%%%%%%%%%%%%%%%%%%%%%%%%%%%%%%%%%%%%%%%%%%%%%%%%%%%%%%%%%%%%%%%%%%%%%%%%%%%%%%%%

%%%%%%%%%%%%%%%%%%%%%%%%%%%%%%%%%%%%%%%%%%%%%%%%%%%%%%%%%%%%%%%%%%%%%%%%%%%%%%%%%%%%%%%%%%%%%%%%%%%%%%%%%%%%%%
\begin{figure*}[!htbp]
\centering

\tikzset{every picture/.style={line width=0.75pt}} %set default line width to 0.75pt        

\begin{tikzpicture}[x=0.75pt,y=0.75pt,yscale=-1,xscale=1]
%uncomment if require: \path (0,145.921875); %set diagram left start at 0, and has height of 145.921875

%Shape: Grid [id:dp9597332067298479] 
\draw  [draw opacity=0] (51,18) -- (111.5,18) -- (111.5,78.48) -- (51,78.48) -- cycle ; \draw   (51,18) -- (51,78.48)(71,18) -- (71,78.48)(91,18) -- (91,78.48)(111,18) -- (111,78.48) ; \draw   (51,18) -- (111.5,18)(51,38) -- (111.5,38)(51,58) -- (111.5,58)(51,78) -- (111.5,78) ; \draw    ;
%Shape: Grid [id:dp1651748896103642] 
\draw  [draw opacity=0] (161,2) -- (261.5,2) -- (261.5,102.48) -- (161,102.48) -- cycle ; \draw   (161,2) -- (161,102.48)(181,2) -- (181,102.48)(201,2) -- (201,102.48)(221,2) -- (221,102.48)(241,2) -- (241,102.48)(261,2) -- (261,102.48) ; \draw   (161,2) -- (261.5,2)(161,22) -- (261.5,22)(161,42) -- (261.5,42)(161,62) -- (261.5,62)(161,82) -- (261.5,82)(161,102) -- (261.5,102) ; \draw    ;
%Shape: Circle [id:dp7724242829636028] 
\draw   (71.5,48.23) .. controls (71.5,42.85) and (75.86,38.48) .. (81.24,38.48) .. controls (86.62,38.48) and (90.98,42.85) .. (90.98,48.23) .. controls (90.98,53.61) and (86.62,57.97) .. (81.24,57.97) .. controls (75.86,57.97) and (71.5,53.61) .. (71.5,48.23) -- cycle ;
%Shape: Circle [id:dp9982013591868288] 
\draw   (201.5,52.47) .. controls (201.5,47.22) and (205.75,42.97) .. (211,42.97) .. controls (216.25,42.97) and (220.5,47.22) .. (220.5,52.47) .. controls (220.5,57.72) and (216.25,61.97) .. (211,61.97) .. controls (205.75,61.97) and (201.5,57.72) .. (201.5,52.47) -- cycle ;
%Shape: Grid [id:dp9173097677304121] 
\draw  [draw opacity=0] (306,33.48) -- (326.5,33.48) -- (326.5,54.48) -- (306,54.48) -- cycle ; \draw   (306,33.48) -- (306,54.48)(326,33.48) -- (326,54.48) ; \draw   (306,33.48) -- (326.5,33.48)(306,53.48) -- (326.5,53.48) ; \draw    ;
%Shape: Grid [id:dp6541433148680511] 
\draw  [draw opacity=0] (306,3.48) -- (326.5,3.48) -- (326.5,24.48) -- (306,24.48) -- cycle ; \draw   (306,3.48) -- (306,24.48)(326,3.48) -- (326,24.48) ; \draw   (306,3.48) -- (326.5,3.48)(306,23.48) -- (326.5,23.48) ; \draw    ;
%Shape: Circle [id:dp016626737428153904] 
\draw   (306.5,13.47) .. controls (306.5,8.22) and (310.75,3.97) .. (316,3.97) .. controls (321.25,3.97) and (325.5,8.22) .. (325.5,13.47) .. controls (325.5,18.72) and (321.25,22.97) .. (316,22.97) .. controls (310.75,22.97) and (306.5,18.72) .. (306.5,13.47) -- cycle ;
%Shape: Grid [id:dp3831054493442876] 
\draw  [draw opacity=0] (306,64.48) -- (326.5,64.48) -- (326.5,85.48) -- (306,85.48) -- cycle ; \draw   (306,64.48) -- (306,85.48)(326,64.48) -- (326,85.48) ; \draw   (306,64.48) -- (326.5,64.48)(306,84.48) -- (326.5,84.48) ; \draw    ;

% Text Node
\draw (82,91.48) node  [align=left] {$\displaystyle k_{2} :K\rightarrow \mathbb{Z}$};
% Text Node
\draw (210,113.48) node  [align=left] {$\displaystyle b:B_2\rightarrow \mathbb{Z}$};
% Text Node
\draw (352,15.48) node  [align=left] {:Origin};
% Text Node
\draw (444,42.48) node  [align=left] {:Value of function is $\displaystyle x$ at the position \ };
% Text Node
\draw (410,73.48) node  [align=left] {:Position not in the domain};
% Text Node
\draw (316,42.48) node  [align=left] {$\displaystyle x$};
% Text Node
\draw (81,49.48) node  [align=left] {0};
% Text Node
\draw (101,68.48) node  [align=left] {10};
% Text Node
\draw (61,29.48) node  [align=left] {10};
% Text Node
\draw (81,29.48) node  [align=left] {10};
% Text Node
\draw (101,29.48) node  [align=left] {10};
% Text Node
\draw (61,49.48) node  [align=left] {10};
% Text Node
\draw (81,69.48) node  [align=left] {10};
% Text Node
\draw (61,68.48) node  [align=left] {10};
% Text Node
\draw (101,49.48) node  [align=left] {10};
% Text Node
\draw (171,33.48) node  [align=left] {10};
% Text Node
\draw (191,53.48) node  [align=left] {20};
% Text Node
\draw (211,73.48) node  [align=left] {20};
% Text Node
\draw (231,93.48) node  [align=left] {10};
% Text Node
\draw (171,53.48) node  [align=left] {10};
% Text Node
\draw (191,73.48) node  [align=left] {20};
% Text Node
\draw (211,93.48) node  [align=left] {10};
% Text Node
\draw (191,93.48) node  [align=left] {10};
% Text Node
\draw (231,73.48) node  [align=left] {20};
% Text Node
\draw (191,13.48) node  [align=left] {10};
% Text Node
\draw (211,33.48) node  [align=left] {20};
% Text Node
\draw (231,53.48) node  [align=left] {20};
% Text Node
\draw (251,73.48) node  [align=left] {10};
% Text Node
\draw (171,73.48) node  [align=left] {10};
% Text Node
\draw (191,33.48) node  [align=left] {20};
% Text Node
\draw (211,13.48) node  [align=left] {10};
% Text Node
\draw (231,33.48) node  [align=left] {20};
% Text Node
\draw (251,53.48) node  [align=left] {10};
% Text Node
\draw (231,13.48) node  [align=left] {10};
% Text Node
\draw (251,33.48) node  [align=left] {10};
% Text Node
\draw (211,53.48) node  [align=left] {10};
% Text Node
\draw (171,12.92) node  [align=left] {10};
% Text Node
\draw (250,12.92) node  [align=left] {10};
% Text Node
\draw (251,91.92) node  [align=left] {10};
% Text Node
\draw (172,91.92) node  [align=left] {10};

\end{tikzpicture}

\caption{The non-flat structuring elements $k_2$ and $b_2$ used in examples}
\label{fig:kb2-grey}
\end{figure*}

%%%%%%%%%%%%%%%%%%%%%%%%%%%%%%%%%%%%%%%%%%%%%%%%%%%%%%%%%%%%%%%%%%%%%%%%%%%%%%%%%%%%%%%%%%%%%%%%%%%%%%%%

\begin{prop}
\label{thm:pb14a}
Let $B\subseteq E^{N}$, and $b:B\rightarrow L$ be the structuring element employed in the dilation and erosion. Then, 
\begin{enumerate}[I.]
\item $(f\vert_S \oplus b\vert_S) \leq (f\oplus b)\vert_S$
\item $(f\vert_S \ominus b\vert_S) \geq (f\ominus b)\vert_S$
\end{enumerate}
\end{prop}

\begin{proof}
\begin{enumerate}[I.]

\item We know, from Proposition \ref{thm:pb14} \textit{I}, 
that $(F\cap S)\oplus (B\cap S) \subseteq (F\oplus B)\cap S$.\\
Let $x\in (F\cap S)\oplus (B\cap S)$. Then, 

$(f\vert_S \oplus b\vert_S)(x)$ 
\begin{align*}
&=\max _{x-u\in F\cap S, \: u\in B\cap S} \{ f\vert_S (x-u) + b\vert_S (u)\}\\
&=\max _{x-u\in F\cap S, \: u\in B\cap S}  \{ f\vert_S (x-u) + b\vert_S (u)\}\\
&\leq\max _{x-u\in F, \: u\in B} \{ f (x-u) + b (u)\}\\
&=(f\oplus b)(x)
\end{align*}
But, $x$ $\in (F\cap S)\oplus (B\cap S)$  $\subseteq$ $(F\oplus B)\cap S$ $\subseteq$  $S$. \\
It follows $(f\oplus b)(x)$ $=$ $(f\oplus b)\vert_S (x)$,\\
which implies $(f\vert_S \oplus b\vert_S)(x) \leq (f\oplus b)\vert_S (x)$.\\
\\
\item We know, from Proposition \ref{thm:pb14} \textit{II}, $(F\ominus B)\cap S$   $\subseteq (F\cap S)\ominus (B\cap S)$.\\ 
Let $x\in (F\ominus B)\cap S$. Then,

\begin{align}
 \label{eq:141}
 \begin{split}
(f\ominus b)\vert_S (x) =& (f\ominus b) (x) \\
=& \min _{u\in B} \{f(x+u)-b(u)\} \\
\leq & \min_{u\in B\cap S} \{f(x+u) - b\vert_S(u)\} 
\end{split}
\end{align}
Thus, $x\in (F\ominus B)\cap S$ and $u\in B\cap S$ $\Rightarrow x+u\in S$ ($\because S\oplus S=S$)\\
Also, $x\in F\ominus B$ and $u\in B$ $\Rightarrow$ $x+u\in F$. \\
Therefore, from \eqref{eq:141}, we have,
\begin{align*}
(f\ominus b)\vert_S (x)  \leq & \min_{u\in B\cap S} \{f(x+u) - b\vert_S(u)\} \\
=& \min_{u\in B\cap S} \{f\vert_S(x+u) - b\vert_S(u)\}\\
=& (f\vert_S \ominus b\vert_S)(x). 
\end{align*}
Thus, $(f\ominus b)\vert _S \leq  (f\vert _S \ominus b\vert _S)$.
\end{enumerate}
\end{proof}

Figures \ref{fig:grey1411} and \ref{fig:grey1412} illustrate the first part of above proposition. We see that dilation in the sampled domain (here, $f\vert _S \oplus b\vert _S$) is bounded above by sampling of dilated image (here, $(f\oplus b)\vert _S$). 
Figures \ref{fig:grey1421} and \ref{fig:grey1422} illustrate the second part of the proposition, i.e.\ erosion in sampled domain (here, $ f\vert _S \ominus b\vert _S$) is bounded below by sampling of eroded image (here, $(f\ominus b)\vert_S $). The illustrations of the above proposition using non-flat SEs $k_2$ and $b_2$ are given in Figures \ref{fig:nonflatgrey1411}, \ref{fig:nonflatgrey1412} \ref{fig:nonflatgrey1421} and \ref{fig:nonflatgrey1422}.

%Proposition \ref{thm:pb14a} is the grayscale version of proposition \ref{thm:pb14}.

\begin{lemma}
\label{thm:lemaa}
\begin{enumerate}[I.]
\item $(f\vert_S \oplus b\vert_S) =\: (f\oplus b\vert_S)\vert_S$
\item $(f\vert_S \ominus b\vert_S) =\: (f\ominus b\vert_S)\vert_S$
\end{enumerate}
\end{lemma}

\begin{proof}

\begin{enumerate}[I.]

\item We know from Lemma \ref{thm:lema} that \\ $[F\oplus (B\cap S)]\cap S =\: (F\cap S)\oplus (B\cap S)$\\
Let $x\in [F\oplus (B\cap S)]\cap S  =\: (F\cap S)\oplus (B\cap S)$. Then,
\begin{align*}
(f\oplus b\vert_S)\vert_S (x) =&(f\oplus b\vert_S)(x)\\
= \max _{u\in B\cap S, \: x-u\in F} &\{f(x-u) + b\vert_S(u\})
\end{align*}
$x\in [F\oplus (B\cap S)]\cap S$ $\Rightarrow x\in S$. Similarly, $u\in B\cap S$ $\Rightarrow  u\in S$\\
$S=\breve{S}$ and $S\oplus S=S$, therefore, $x-u\in S$
Thus, we have, $(f\oplus b\vert_S)\vert_S (x) =$
\begin{align*}
&\max _{u\in B\cap S, \: x-u\in F}\{f(x-u) + b\vert_S(u\})= \\
&\max _{u\in B\cap S, \: x-u\in F\cap S}\{f(x-u) + b\vert_S(u\}) =\\
&\max _{u\in B\cap S, \: x-u\in F\cap S}\{f\vert_S(x-u) + b\vert_S(u\}) = \\
&(f\vert_S \oplus b\vert_S)(x) 
\end{align*}
This is true for each $x\in  [F\oplus (B\cap S)]\cap S$  $=(F\cap S)\oplus (B\cap S)$ ,  therefore $(f\vert_S \oplus b\vert_S) =\: (f\oplus b\vert_S)\vert_S$.\\
\\

\item We know from Lemma \ref{thm:lema} that \\ $[F\ominus (B\cap S)]\cap S = \: (F\cap S)\ominus (B\cap S)$.\\
Let $x\in F\ominus (B\cap S)]\cap S =\: (F\cap S)\ominus (B\cap S)$. Then, 
\begin{equation*}
(f\ominus b\vert_S)(x)=\: \min _{u\in B\cap S} \{f(x+u) -b\vert_S(u)\}
\end{equation*}
$x\in [F\ominus (B\cap S)]\cap S$ and $u\in (B\cap S)$ $\Rightarrow x\in S,\: u\in S,\: x+u\in S$ and $x+u\in F$ i.e $x+u\in F\cap S$. Therefore, $(f\ominus b\vert_S)(x)=$
\begin{align*}
&\min _{u\in B\cap S} \{f(x+u) -b\vert_S(u)\}=\\
&\min _{u\in B\cap S} \{f\vert_S(x+u) -b\vert_S(u)\}=\\ 
&(f\vert_S \ominus b\vert_S)(x)
\end{align*}
This holds for all $x\in  [F\ominus (B\cap S)]\cap S =\: (F\cap S)\ominus (B\cap S)$.Thus,  $(f\vert_S \ominus b\vert_S) =\: (f\ominus b\vert_S)\vert_S$.

\end{enumerate}
\end{proof}

\begin{lemma}
\label{thm:lemba}
Let $B=B\circ K$ and $b=b\circ k$ Then,
$(f\vert_S\bullet k)\oplus b\leq \: (f\vert_S \oplus k)\oplus b\vert_S $
\end{lemma}

\begin{proof}
By Umbra Homomorphism Theorem \ref{thm:58}, domain of $(f\vert_S\bullet k)$ is $(F\cap S)\bullet K$.\\
Let $x\in [(F\cap S)\bullet K]\oplus B \subseteq \: [(F\cap S)\oplus K]\oplus (B\cap S) $ 
(by Lemma \ref{thm:lemb}). Then, 
\begin{equation*}
(f\vert_S\bullet k)(x) = \max _{x-u\in (F\cap S)\bullet K, \: u\in B }\{(f\vert_S\bullet k)(x-u) +b(u)\}
\end{equation*}
For each $u\in B$, $B=B\circ K$ $\therefore \exists y$ such that $K_y\subset B$ and by Sampling Conditions (see, Theorem \ref{thm:b6}, \ref{cond:g4}), $\exists z\in K_y\cap K_u \cap S$.\\
$z=k_0 +u$ for some $k_0 \in K$. Also, $z\in K_y \subseteq B$ $\Rightarrow z\in B\cap S$.\\
By Proposition \ref{thm:pb22}, $b(u)\leq b(z)+k(-k_0)$.\\
Since, $K=\breve{K}$ , $(x-u) \in (F\cap S)\bullet K \Rightarrow \: (x-u)-k_0 \in (F\cap S)\oplus K$ and
\begin{align*} 
&(f\vert_S \bullet k)(x-u)=\\
&((f\vert_S \oplus k)\ominus k)(x-u)=\\
&\min _{a\in K} \{(f\vert_S \oplus k)(x-u+a)-k(a)\} \leq \\
&(f\vert_S \oplus k)(x-u-k_0)-k(-k_0)
\end{align*}
For each $u\in B$ satisfying $(x-u) \in  (F\cap S)\bullet K\:  \exists z=(u+k_0) \in (B\cap S)$ satisfying $(x-z) \in (F\cap S)\oplus K$ and thus, we have, 
\begin{align*}
&(f\vert_S\bullet k)(x-u)+b(u) \leq  \\
&(f\vert_S \oplus k)(x-z) -k(-k_0) +b(z) +k(k_0) =\\ &(f\vert_S \oplus k)(x-z) + b\vert_S(z)
\end{align*}
Thus, $(f\vert_S\bullet k)(x)=$
\begin{align*} 
&\max _{x-u\in (F\cap S)\bullet K, \: u\in B }\{(f\vert_S\bullet k)(x-u) +b(u)\} \\
&\leq \: \max_{(x-z)\in (F\cap S)\oplus K,\: z\in (B\cap S)  }\{(f\vert_S \oplus k)(x-z) + b\vert_S(z)\}\\ 
&= ((f\vert_S \oplus k)\oplus b\vert_S)(x)
\end{align*}
This holds for all $x\in$  $[(F\cap S)\bullet K]\oplus B$ $\subseteq$  $[(F\cap S)\oplus K]\oplus (B\cap S)$. Therefore, $(f\vert_S\bullet k)\oplus b$ $\leq  (f\vert_S \oplus k)\oplus b\vert_S$.

\end{proof}

We arrive at two of the main results of  this section. The following two theorems illustrate the interaction between sampling and grey-value dilation and grey-value erosion. %optional line.

\begin{theorem} {Grey-value Sample Dilation Theorem }
\label{thm:2ba}
Let $B=B\circ K$ and $b=b\circ k$, then $f\vert_S \oplus b\vert_S =\: ((f\vert_S\bullet k)\oplus b)\vert_S$.  
\end{theorem}

\begin{proof}
By Lemma \ref{thm:lemaa}, $f\vert_S \oplus b\vert_S =\: (f\vert_S \oplus b)\vert_S$.\\ 
By Extensivity property of closing (\textit{Proposition 68} of \cite{r1}),  $f\vert_S \leq \: f\vert_S \bullet k$ $\Rightarrow (f\vert_S\oplus b)\vert_S \leq ((f\vert_S \bullet k)\oplus b)\vert_S$\\
$\Rightarrow f\vert_S \oplus b\vert_S \leq \: ((f\vert_S \bullet k)\oplus b)\vert_S$.\\
\\
Let $ x\in [((F\cap S)\bullet K)\oplus B]\cap S \subseteq \:  [(F\cap S)\oplus K]\oplus (B\cap S) $ (Lemma \ref{thm:lemb}).\\
For any $u\in B\cap S$, we have $(x-u) \in S$ $\because S=S\oplus S$ and $S=\breve{S}$. Also, by the Sampling conditions, $[(F\cap S)\oplus K]\cap S =\: F\cap S$ and $(f\vert_S \oplus k)\vert_S = f\vert_S$.
From Lemma \ref{thm:lemba}, we have 
\begin{align*}
& ((f\vert_S \bullet k)\oplus b)\vert_S (x)  \leq\:  ((f\vert_S\oplus k)\oplus b\vert_S)\vert_S(x) \\
& =\:((f\vert_S\oplus k)\oplus b\vert_S)(x) \\
& = \: \max _{(x-u)\in (F\cap S)\oplus K ,\: u\in B\cap S} \{(f\vert_S \oplus k)(x-u) + b\vert_S (u) \} \\
& =\: \max _{(x-u)\in [(F\cap S)\oplus K]\cap S ,\: u\in B\cap S} \{(f\vert_S \oplus k)\vert_S(x-u)+ b\vert_S (u) \} \\
& =\: \max _{(x-u)\in (F\cap S) ,\: u\in B\cap S}  \{(f\vert_S)(x-u) + b\vert_S (u) \} \\
& =\: (f\vert_S \oplus b\vert _S)(x)
\end{align*}
i.e.\ $((f\vert_S \bullet k)\oplus b)\vert_S (x)$ $\leq (f\vert_S \oplus b\vert _S)(x)$,\\
$\forall x$ $\in [((F\cap S)\bullet K)\oplus B]\cap S$. \\
\\
Thus, $(f\vert_S \bullet k)\oplus b)\vert_S  \leq (f\vert_S \oplus b\vert_S)$\\
i.e.\ $f\vert_S \oplus b\vert_S =\: ((f\vert_S\bullet k)\oplus b)\vert_S$.
\end{proof}

Figures \ref{fig:grey1411} and \ref{fig:grey22} are identical. This illustrates the Grey-value Sample Dilation Theorem, i.e.\ dilation in sampled domain (here, $ (f\vert_S \oplus b\vert_S)$) is equivalent to  sampling after dilating the minimal reconstruction (here, $((f\vert_S\bullet k)\oplus b)\vert_S $). The theorem is illustrated by Figures  \ref{fig:nonflatgrey1411} and \ref{fig:nonflatgrey22} for non-flat morphology.

\begin{figure*}[!htbp]
\minipage{0.24\linewidth}
 \includegraphics[width=\linewidth]{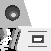}
  \caption{$ (f\oplus b)\vert_S$}\label{fig:grey1412}
\endminipage\hfill
\minipage{0.24\linewidth}
  \includegraphics[width=\linewidth]{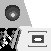}
  \caption{$(f\vert_S \oplus b\vert_S)$}\label{fig:grey1411}
\endminipage\hfill
\minipage{0.24\linewidth}
 \includegraphics[width=\linewidth]{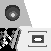}
  \caption{$((f\vert_S\bullet k)\oplus b)\vert_S$}\label{fig:grey22}
\endminipage\hfill
\end{figure*}

\begin{figure*}[!htbp]
\minipage{0.24\linewidth}
 \includegraphics[width=\linewidth]{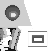}
  \caption{$ (f\oplus b_2)\vert_S$}\label{fig:nonflatgrey1412}
\endminipage\hfill
\minipage{0.24\linewidth}
  \includegraphics[width=\linewidth]{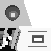}
  \caption{$(f\vert_S \oplus b_2 \vert_S)$}\label{fig:nonflatgrey1411}
\endminipage\hfill
\minipage{0.24\linewidth}
 \includegraphics[width=\linewidth]{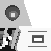}
  \caption{$((f\vert_S\bullet k_2)\oplus b_2 )\vert_S$}\label{fig:nonflatgrey22}
\endminipage\hfill
\end{figure*}

\begin{remark}
It can be shown using Umbra Homomorphism Theorem \ref{thm:58} that opening and closing of grey-value images are increasing operations. That is, if $f\leq g$, then for any S.E $c:C\rightarrow L$, $f\circ c\leq g\circ c$ and $f\bullet c \leq g \bullet c$. 
\end{remark}

\begin{theorem}{Grey-value Sampling Erosion Theorem}
\label{thm:3ba}
Let $B=B\circ K$  and $b=b\circ k$. 
Then $f\vert_S \ominus b\vert_S = ((f\vert_S \oplus k)\ominus b)\vert_S$.
\end {theorem}

\begin{proof}
From Grey-value Morphological Sampling Theorem \ref{thm:b6}, we have $(f\vert_S \oplus k)\vert_S =\: f\vert_S$,\\
i.e.\ $(f\vert_S \ominus b\vert_S)=\: (f\vert_S \oplus k)\vert_S \ominus b\vert_S$.\\
From Lemma \ref{thm:lemaa}, we have  $(f\vert_S \oplus k)\vert_S \ominus b\vert_S=\: ((f\vert_S \oplus k)\ominus b\vert_S)\vert_S \geq \: ((f\vert_S\oplus k)\ominus b)\vert_S$.\\
\\
We show that $f\vert_S \ominus b\vert_S \leq \: ((f\vert_S \oplus k)\ominus b)\vert_S$.\\
Let $x\in (F\cap S)\ominus (B\cap S) =\{ [(F\cap S)\oplus K]\ominus B\}\cap S $ (by Theorem \ref{thm:b3}).\\
We show that $\min _{u\in B\cap S} \{(f\vert_S)(x+u) -b\vert_S(u)\} \leq \min _{u' \in B} \{((f\vert_S)\oplus k)(x+u') -b(u')\}$\\
$u'\in B \subseteq (B\cap S) \oplus K $
$\Rightarrow u'$ $=u_0+ k_0$, where $u_0 \in B\cap S$ and $k_0 \in K $\\
$b=b\circ k$, therefore, by Proposition \ref{thm:pb22}, we have $b(u')\leq b\vert_S(u_0) + k(k_0)$.\\
If $u'=u_0+k_0$ , then $((f\vert_S)\oplus k)(x+u')\geq \: (f\vert_S)(x+u_0) +k(k_0)$.\\
Thus, for each $u'\in B$,  $\exists u_0\in B\cap S$ such that 
$ \{((f\vert_S)\oplus k)(x+u') -b(u')\} \geq \: \{(f\vert_S)(x+u_0) -b\vert_S(u_0)\}$\\
i.e. $\min _{u\in B\cap S} \{(f\vert_S)(x+u) -b\vert_S(u)\}\leq$ $\min _{u' \in B} \{((f\vert_S)\oplus k)(x+u') -b(u')\}$, \\
 $\forall x$ $\in (F\cap S)\ominus (B\cap S)$.
\\
 $\Rightarrow \: f\vert_S \ominus b\vert_S \leq \: ((f\vert_S \oplus k)\ominus b)\vert_S$.
 
\end{proof}

As expected, Figure \ref{fig:grey1421}, erosion in sampled domain (here, $(f\vert_S \ominus b\vert_S)$), is identical to Figure \ref{fig:grey32}, sampling after eroding the maximal reconstruction (here, $(f\vert_S \oplus k)\vert_S \ominus b\vert_S$ ), thus demonstrating an example of the above theorem.
Figures \ref{fig:nonflatgrey1421} and \ref{fig:nonflatgrey32}  illustrate the above theorem for non-flat filter and SE.

%Image Rendering Error:
%some rendering error in the right images of next two figures. The images render differently for middle and right figure, even if the same image file is used. The middle figure is as expected.

\begin{figure*}[!htbp]
\minipage{0.24\linewidth}
 \includegraphics[width=\linewidth]{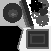}
  \caption{$ (f\ominus b)\vert_S $}\label{fig:grey1422}
\endminipage\hfill
\minipage{0.24\linewidth}
  \includegraphics[width=\linewidth]{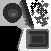}
  \caption{$(f\vert_S \ominus b\vert_S)$ }\label{fig:grey1421}
\endminipage\hfill
\minipage{0.24\linewidth}
 \includegraphics[width=\linewidth]{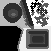}
  \caption{$(f\vert_S \oplus k)\vert_S \ominus b\vert_S$}\label{fig:grey32}
\endminipage\hfill
\end{figure*}

\begin{figure*}[!htbp]
\minipage{0.24\linewidth}
 \includegraphics[width=\linewidth]{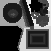}
  \caption{$ (f\ominus b_2)\vert_S$ \textcolor{white}{ABCDE}}\label{fig:nonflatgrey1422}
\endminipage\hfill
\minipage{0.24\linewidth}
  \includegraphics[width=\linewidth]{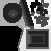}
  \caption{$(f\vert_S \ominus b_2 \vert_S)$ \textcolor{white}{ABCDE}}\label{fig:nonflatgrey1421}
\endminipage\hfill
\minipage{0.24\linewidth}
 \includegraphics[width=\linewidth]{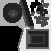}
  \caption{$(f\vert_S \oplus k_2 )\vert_S \ominus b_2 \vert_S$}\label{fig:nonflatgrey32}
\endminipage\hfill
\end{figure*}

We now proceed to discuss the interaction of grey-value opening and closing with sampling and reconstruction.

\begin{prop}
\label{thm:pb16a}
$(f\circ b\vert_S)\vert_S =f\vert_S\circ b\vert_S$
\end{prop}

\begin{proof}
Let $x\in (F\circ (B\cap S)\cap S =(F\cap S)\circ (B\cap S)$ (by Proposition \ref{thm:pb16}).\\
Clearly, $f\vert_S\circ b\vert_S \leq (f\circ b\vert_S)$ $\Rightarrow (f\vert_S\circ b\vert_S) (x) \leq (f\circ b\vert_S)(x) =((f\circ b\vert_S)\vert_S)(x)$\\
Thus, $(f\circ b\vert_S)\vert_S \geq \: f\vert_S\circ b\vert_S$.\\
\\
We show that if  $x\in (F\circ (B\cap S))\cap S =(F\cap S)\circ (B\cap S)$ then $((f\circ b\vert_S)\vert_S)(x) \leq \: (f\vert_S\circ b\vert_S)(x)$\\
\\
$(f\circ b\vert_S)(x) =T(U[f]\circ U[b\vert_S])(x) =$\\ 
$T(\bigcup _{\{(y,y_0)\vert U[b\vert_S]_{(y,y_0)}\subseteq U[f],\: y_0\geq 0 \}} U[b\vert_S]_{(y,y_0)})(x)$\\
$\Rightarrow (x, (f\circ b\vert_S)(x))\in U[b\vert_S]_{(y,y_0)}\subseteq U[f]$\\
$\Rightarrow x=b+y$ for some $b\in B\cap S$ $x,b\in S$ implies $y\in S$ and\\ $(B\cap S)_y \subseteq F \Rightarrow (B\cap S)_y \subseteq (F\cap S)$\\
Taking intersection with $(S\times L)$, we have,\\
$(x, (f\circ b\vert_S)(x))\in\: U[b\vert_S]_{(y,y_0)}\cap (S\times L)=$ \\
$U[b\vert_S]_{(y,y_0)} \subseteq\: U[f]\cap (S\times L) =\: U[f\vert_S]$\\
$\Rightarrow\: (x, (f\circ b\vert_S)(x)) \in \: U[b\vert_S]_{(y,y_0)} \subseteq \: U[f\vert_S]$\\
$\Rightarrow \:   (x, (f\circ b\vert_S)(x)) \in \: U[f\vert_S \circ b\vert_S]$\\
$\Rightarrow \: ((f\circ b\vert_S)\vert_S)(x) \leq \: (f\vert_S\circ b\vert_S)(x)$\\
Thus, $(f\circ b\vert_S)\vert_S \leq \: f\vert_S\circ b\vert_S$\\
$\therefore (f\circ b\vert_S)\vert_S =f\vert_S\circ b\vert_S$
\end{proof}

\begin{prop}
\label{thm:pb17a}
$(f\bullet b\vert_S) =f\vert_S\bullet b\vert_S$
\end{prop}

\begin{proof}
Let $x\in (F\cap S)\bullet (B\cap S) = (F\bullet (B\cap S))\cap S$ (by Proposition \ref{thm:pb17}).\\
Clearly, since closing is increasing and $f\vert _S \leq f$, we have,$(f\vert_S \bullet b\vert_S)(x) \leq \: (f\bullet b\vert_S)(x)  =\: ((f\bullet b\vert_S)\vert_S) (x) $. \\
Therefore, $(f\vert_S \bullet b\vert_S) \leq \:  ((f\bullet b\vert_S)\vert_S) $.\\
\\
We show that if $x\in  (F\cap S)\bullet (B\cap S)$ $= (F\bullet (B\cap S))\cap S$,  $\Rightarrow$ $(x,((f\bullet b\vert_S)\vert_S)(x)) =\: (x,(f\bullet b\vert_S)(x))$ $\in U[f\vert_S \bullet b\vert_S]$.\\
\\ 
$(x,(f\bullet b\vert_S)(x)) \in U[(f\bullet b\vert_S)\vert_S] = \: U[(f\bullet b\vert_S)] \cap (S\times L)$\\
$\Rightarrow\:  \exists (y,y_0) \:, y_0\geq 0$ such that $(x,(f\bullet b\vert_S)(x)) \in \tilde{U[b\vert_S]_{(y,y_0)}}$ and $\tilde{U[b\vert_S]_{(y,y_0)}} \cap U[f] \neq \emptyset$\\
$\Rightarrow x\in \breve{(B\cap S)_y}$ and $ \breve{(B\cap S)_y}\cap F \neq \emptyset$\\
$\Rightarrow y\in S$ and $ \breve{(B\cap S)_y} \cap F =  \breve{(B\cap S)_y}\cap (F\cap S)$\\
$\Rightarrow \tilde{U[b\vert_S]_{(y,y_0)}} \cap (S\times L)  =  \tilde{U[b\vert_S]_{(y,y_0)}} $.\\
% $\because x\in S$ and x= -s +y$ for some $s\in (B\cap S)$ $S=\breve{S}$ and $S=S\oplus S$;  and thus we also have $\breve{(B\cap S)_y} \subseteq S
Therefore, $(x,((f\bullet b\vert_S)\vert_S)(x)) =\: (x,(f\bullet b\vert_S)(x))  \in$\\  $\tilde{U[b\vert_S]_{(y,y_0)}} \cap (S\times L) =  \tilde{U[b\vert_S]_{(y,y_0)}} $ and \\
$\tilde{U[b\vert_S]_{(y,y_0)}} \cap (S\times L) \cap U[f]=\: \tilde{U[b\vert_S]_{(y,y_0)}} \cap U[f\vert_S] \neq \emptyset $\\
 $\Rightarrow \: (x,((f\bullet b\vert_S)\vert_S)(x))\in U[f\vert_S\bullet b\vert_S]$\\
 $\Rightarrow \:(f\bullet b\vert_S)\vert_S \leq \: f\vert_S\bullet b\vert_S$,\\
 \\
 $\therefore (f\bullet b\vert_S)\vert_S =f\vert_S\bullet b\vert_S$.
\end{proof}

%Optional Line:
The next theorem is another major result of this section. The next theorem bounds opening and closing in sampled domain, by sampling after opening or closing. 

\begin{theorem} {Grey-value Sample Opening and Closing Bounds Theorem}
\label{thm:b4a}
Let $B=B\circ K$ and $b=b\circ k$, then 
\begin{enumerate}[I.]
\item $(f\circ [b\vert_S\oplus k])\vert_S \leq \: f\vert_S\circ b\vert_S \leq \: ((f\vert_S\oplus k)\circ b)\vert_S$
\item $((f\vert_S \bullet k)\bullet b)\vert_S \leq \:f\vert_S \bullet b\vert_S \leq \: (f\bullet (b\vert_S \oplus k))\vert_S$
\end{enumerate}
\end{theorem}

\begin{proof}
\textit{I}\\
It is shown in Theorem \ref{thm:b4} that $[F\circ (B\cap S)] \supseteq \: [F\circ ((B\cap S)\oplus K)]$.\\
Using Umbra Homomorphism Theorem \ref{thm:58}, we have $U[f\circ b\vert_S] =U[f]\circ U[b\vert_S] \supseteq U[f]\circ U[b\vert_S \oplus k] =U[f\circ (b\vert_S \oplus k)]$.\\
Using Proposition \ref{thm:pb16a}, $U[f\vert_S \circ b\vert_S]  \: = \: U[(f\circ b\vert_S)\vert_S] = \: U[f\circ b\vert_S] \cap (S\times L) \supseteq  U[f\circ (b\vert_S \oplus k)]\cap (S\times L) = \: U[(f\circ (b\vert_S \oplus k))\vert_S]$\\
$\Rightarrow (f\circ [b\vert_S\oplus k])\vert_S \leq \: f\vert_S\circ b\vert_S $.\\
\\

Let $x\in (F\cap S)\circ (B\cap S)$.\\
We show that  $(x, (f\vert_S \circ b\vert_S)(x)) \in U[((f\vert_S \oplus k)\circ b)\vert_S]$.\\
$U[f\vert_S \circ b\vert_S] = U[f\vert_S]\circ U[b\vert_S]$\\
$(x,(f\vert_S\circ b\vert_S)(x)) \in U[f\vert_S \circ b\vert_S] \Rightarrow$
$\exists (y,y_0)\in E^N \times L$, $ y_0 \geq 0$,
such that \\
 $(x,(f\vert_S\circ b\vert_S)(x)) \in U[b\vert_S]_{(y,y_0)} \subseteq U[f\vert_S]$.\\
We have, $U[b\vert_S] =\: U[b]\cap (S\times L) \subseteq U[b]$ $\Rightarrow \: (x,(f\vert_S\circ b\vert_S)(x)) \in U[b]$. \\
Since, $b=b\circ k$, $b\vert_S \oplus k \geq b$. Therefore, \\
$U[b\vert_S\oplus k]_{(y,y_0)} \supseteq U[b]_{(y,y_0)}$.\\
We show that $U[b\vert_S\oplus k]_{(y,y_0)}  \subseteq U[f\vert_S \oplus k]$.\\
For any $u\in [(B\cap S)_y\oplus K]$, $(b\vert_S \oplus k)(u-y) +y_0 =$
  \begin{align*}
 &\: \max _{u-s\in K, s\in (B\cap S)_y} \{b\vert_S(s-y) + k(u-s)\} +y_0\\
 &=\: (b\vert_S)(s_0 -y) +k(u-s_0) +y_0 \\
 &\text{\textcolor{white}{ABCDEF} for some } s_0 \in (B\cap S)_y\\
 &\leq \: (f\vert_S)(s_0) +k(u-s_0) \\
 & \text{\textcolor{white}{ABCDEF}}  \because U[b\vert_S]_{(y,y_0)} \subseteq U[f\vert_S]\\
 & \leq  \: (f\vert_S \oplus k)(u)\\
\end{align*} 
$\Rightarrow U[b\vert_S\oplus k]_{(y,y_0)}  \subseteq U[f\vert_S \oplus k]$ \\
$\Rightarrow (x,(f\vert_S\circ b\vert_S)(x)) \in \: U[b]_{(y,y_0)} \subseteq U[f\vert_S \oplus k]$\\
$\Rightarrow  (x,(f\vert_S\circ b\vert_S)(x))  \in \: U[(f\vert_S \oplus k)\circ b]$.\\
Since, $x\in S$, $(x,(f\vert_S\circ b\vert_S)(x))  \in \: U[(f\vert_S \oplus k)\circ b] \cap (S\times L) =U[((f\vert_S \oplus k)\circ b)\vert_S]$\\
$\Rightarrow f\vert_S \circ b\vert_S \leq \: ((f\vert_S \oplus k)\circ b)\vert_S$.\\
\\
\textit{II} \\
We know, by Theorem \ref{thm:b4} that, $\{ [(F \cap S)\bullet K]\bullet B\} \cap S \subseteq \: (F\cap S)\bullet (B\cap S) \subseteq \: \{F\bullet [(B\cap S)\oplus K] \} \cap S$.\\
By Proposition \ref{thm:pb17a}, we have $f\vert_S \bullet b\vert_S = (f\bullet b\vert_S)\vert_S$.\\
And, from Umbra Homomorphism Theorem \ref{thm:58} and Proposition \ref{thm:refUm}, we have  $(f\bullet b\vert_S)\vert_S \leq (f\bullet (b\vert_S \oplus k))\vert_S$,
i.e.\  if $x \in (F\bullet (B\cap S))\cap S$, then $\exists (y,y_0)\in E^N\times L$,  $y_0 \geq 0$ such that $(x, (f\bullet b\vert_S)(x)) \in \tilde{U[b\vert_S]_{(y,y_0)}}$ and  $ \tilde{U[b\vert_S]_{(y,y_0)}} \cap U[f] \neq \emptyset $\\
 $\tilde{U[b\vert_S]_{(y,y_0)} } \subseteq \tilde{U[b\vert_S \oplus k]_{(y,y_0)}}$ $\Rightarrow  (x, (f\bullet b\vert_S)(x)) \in  \tilde{U[b\vert_S \oplus k]_{(y,y_0)}}$ and $ \tilde{U[b\vert_S \oplus k]_{(y,y_0)}} \cap U[f] \neq \emptyset$\\
 $\Rightarrow  (x, (f\bullet b\vert_S)(x)) \in U[(f\bullet (b\vert_S\oplus k))\vert_S]$ ($\because x\in S$),\\
 i.e.\ $f\vert_S \bullet b\vert_S = \: (f\bullet b\vert_S)\vert_S \leq \: (f\bullet (b\vert_S \oplus k))\vert_S$.\\
 \\
 \\
 
We know from Theorem \ref{thm:b4}, that under given conditions, $\{[(F\cap S)\bullet K]\bullet B\}\cap S \subseteq (F\cap S)\bullet (B\cap S)$.\\
Let $r=\: (f\vert_S \bullet k)$.\\
We, first prove $(r\oplus b)\circ k=r\oplus b$. \\
Opening of Grey-value image is anti-extensive (\textit{Proposition 67 } of \cite{r1}). Therefore, $(r\oplus b)\circ k \leq\:  r\oplus b$.\\
We show  $r\oplus b \leq (r\oplus b)\circ k$, i.e.\ $U[r\oplus b] \subseteq U[(r\oplus b)\circ k]$.\\

\begin{align*}
& U[(r\oplus b)\circ k] \\
 & = \bigcup _{\{(y,y_0)\in E^N\times L \vert U[k]_{(y,y_0)} \subseteq \bigcup _{(a,b)\in U[r] } U[b]_{(a,b)}\} } U[k]_{(y,y_0)}\\
& \supseteq\: \bigcup _{(a,b)\in U[r]} \{ \bigcup _{\{ (y,y_0)\in E^N \times L \vert U[k]_{(y,y_0)} \subseteq U[b]_{(a,b)}  \}} U[k]_{(y,y_0)}  \} \\
&=  \bigcup _{(a,b)\in U[r]} U[b]_{(a,b)} \\
& = U[b\oplus r]\\
& \text{i.e.\  } r\oplus b  = (r\oplus b)\circ k\\
\end{align*}

By Grey-value Sampling Theorem \ref{thm:b6}, we have $r\oplus b \leq \: ((r\oplus b)\vert_S \oplus k)$.\\
$(r\bullet  b)\vert_S =\: ((r\oplus b)\ominus b)\vert_S \leq (((r\oplus b)\vert_S\oplus k)\ominus b)\vert_S$.\\
By Theorem \ref{thm:3ba}, ($\because b=b\circ k$), we have , $(((r\oplus b)\vert_S \oplus k)\ominus b)\vert_S =\: (r\oplus b)\vert_S \ominus b\vert_S$.\\
By Theorem \ref{thm:2ba}, ($\because r=(f\vert_S \bullet k)$), we have   $(r\oplus b)\vert_S \ominus b\vert_S =\: (f\vert_S \oplus b\vert_S)\ominus b\vert_S =\: f\vert_S \bullet b\vert_S$\\
$\Rightarrow \: (r\bullet b)\vert_S\leq f\vert_S \bullet b\vert_S$\\
i.e.\ $((f\vert_S\bullet k)\bullet b)\vert_S \leq f\vert_S \bullet b\vert_S $. 
\end{proof}

\begin{figure*}[!htbp]
\minipage{0.24\linewidth}
 \includegraphics[width=\linewidth]{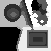}
  \caption{$(f\circ [b\vert_S\oplus k])\vert_S $}\label{fig:grey411}
\endminipage\hfill
\minipage{0.24\linewidth}
  \includegraphics[width=\linewidth]{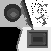}
  \caption{$ \:f\vert_S\circ b\vert_S$}\label{fig:grey412}
\endminipage\hfill
\minipage{0.24\linewidth}
 \includegraphics[width=\linewidth]{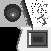}
  \caption{$((f\vert_S\oplus k)\circ b)\vert_S$}\label{fig:grey413}
\endminipage\hfill
\end{figure*}

\begin{figure*}[!htbp]
\minipage{0.24\linewidth}
 \includegraphics[width=\linewidth]{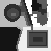}
  \caption{$(f\circ [b_2 \vert_S\oplus k_2 ])\vert_S $}\label{fig:nonflatgrey411}
\endminipage\hfill
\minipage{0.24\linewidth}
  \includegraphics[width=\linewidth]{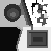}
  \caption{$ \:f\vert_S\circ b_2 \vert_S$}\label{fig:nonflatgrey412}
\endminipage\hfill
\minipage{0.24\linewidth}
 \includegraphics[width=\linewidth]{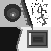}
  \caption{$((f\vert_S\oplus k_2 )\circ b)\vert_S$}\label{fig:nonflatgrey413}
\endminipage\hfill
\end{figure*}

\begin{figure*}[h!]
\minipage{0.24\linewidth}
 \includegraphics[width=\linewidth]{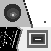}
  \caption{$((f\vert_S \bullet k)\bullet b)\vert_S  $}\label{fig:grey421}
\endminipage\hfill
\minipage{0.24\linewidth}
  \includegraphics[width=\linewidth]{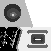}
  \caption{$ f\vert_S \bullet b\vert_S $}\label{fig:grey422}
\endminipage\hfill
\minipage{0.24\linewidth}
 \includegraphics[width=\linewidth]{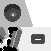}
  \caption{$ (f\bullet (b\vert_S \oplus k))\vert_S$}\label{fig:grey423}
\endminipage\hfill
\end{figure*}

\begin{figure*}[h!]
\minipage{0.24\linewidth}
 \includegraphics[width=\linewidth]{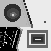}
  \caption{$((f\vert_S \bullet k_2)\bullet b_2)\vert_S $}\label{fig:nonflatgrey421}
\endminipage\hfill
\minipage{0.24\linewidth}
  \includegraphics[width=\linewidth]{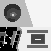}
  \caption{$f\vert_S \bullet b_2 \vert_S $}\label{fig:nonflatgrey422}
\endminipage\hfill
\minipage{0.24\linewidth}
 \includegraphics[width=\linewidth]{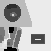}
  \caption{$(f\bullet (b_2 \vert_S \oplus k_2 ))\vert_S$}\label{fig:nonflatgrey423}
\endminipage\hfill
\end{figure*}

Figures \ref{fig:grey411}-\ref{fig:grey413} demonstrate interaction of sampling with opening operation.  Figures \ref{fig:grey421}-\ref{fig:grey423} demonstrate interaction of sampling with closing operation.

%%%%%%%%%%%%%%%%%%%%%%%%%%%%%%%%%%%%%%%%%%%%%%%%%%%%%%%%%%%%%%%%%%%%%%%%%%%%%%%%%%%%%%%%%%%%%%%%%%%%%%%%%%
\begin{figure*}[!htbp]
\centering

\tikzset{every picture/.style={line width=0.75pt}} %set default line width to 0.75pt        

\begin{tikzpicture}[x=0.75pt,y=0.75pt,yscale=-1,xscale=1]
%uncomment if require: \path (0,145.921875); %set diagram left start at 0, and has height of 145.921875

%Shape: Grid [id:dp9597332067298479] 
\draw  [draw opacity=0] (51,18) -- (111.5,18) -- (111.5,78.48) -- (51,78.48) -- cycle ; \draw   (51,18) -- (51,78.48)(71,18) -- (71,78.48)(91,18) -- (91,78.48)(111,18) -- (111,78.48) ; \draw   (51,18) -- (111.5,18)(51,38) -- (111.5,38)(51,58) -- (111.5,58)(51,78) -- (111.5,78) ; \draw    ;
%Shape: Grid [id:dp1651748896103642] 
\draw  [draw opacity=0] (161,2) -- (261.5,2) -- (261.5,102.48) -- (161,102.48) -- cycle ; \draw   (161,2) -- (161,102.48)(181,2) -- (181,102.48)(201,2) -- (201,102.48)(221,2) -- (221,102.48)(241,2) -- (241,102.48)(261,2) -- (261,102.48) ; \draw   (161,2) -- (261.5,2)(161,22) -- (261.5,22)(161,42) -- (261.5,42)(161,62) -- (261.5,62)(161,82) -- (261.5,82)(161,102) -- (261.5,102) ; \draw    ;
%Shape: Circle [id:dp7724242829636028] 
\draw   (71.5,48.23) .. controls (71.5,42.85) and (75.86,38.48) .. (81.24,38.48) .. controls (86.62,38.48) and (90.98,42.85) .. (90.98,48.23) .. controls (90.98,53.61) and (86.62,57.97) .. (81.24,57.97) .. controls (75.86,57.97) and (71.5,53.61) .. (71.5,48.23) -- cycle ;
%Shape: Circle [id:dp9982013591868288] 
\draw   (201.5,52.47) .. controls (201.5,47.22) and (205.75,42.97) .. (211,42.97) .. controls (216.25,42.97) and (220.5,47.22) .. (220.5,52.47) .. controls (220.5,57.72) and (216.25,61.97) .. (211,61.97) .. controls (205.75,61.97) and (201.5,57.72) .. (201.5,52.47) -- cycle ;
%Shape: Grid [id:dp9173097677304121] 
\draw  [draw opacity=0] (306,33.48) -- (326.5,33.48) -- (326.5,54.48) -- (306,54.48) -- cycle ; \draw   (306,33.48) -- (306,54.48)(326,33.48) -- (326,54.48) ; \draw   (306,33.48) -- (326.5,33.48)(306,53.48) -- (326.5,53.48) ; \draw    ;
%Shape: Grid [id:dp6541433148680511] 
\draw  [draw opacity=0] (306,3.48) -- (326.5,3.48) -- (326.5,24.48) -- (306,24.48) -- cycle ; \draw   (306,3.48) -- (306,24.48)(326,3.48) -- (326,24.48) ; \draw   (306,3.48) -- (326.5,3.48)(306,23.48) -- (326.5,23.48) ; \draw    ;
%Shape: Circle [id:dp016626737428153904] 
\draw   (306.5,13.47) .. controls (306.5,8.22) and (310.75,3.97) .. (316,3.97) .. controls (321.25,3.97) and (325.5,8.22) .. (325.5,13.47) .. controls (325.5,18.72) and (321.25,22.97) .. (316,22.97) .. controls (310.75,22.97) and (306.5,18.72) .. (306.5,13.47) -- cycle ;
%Shape: Grid [id:dp3831054493442876] 
\draw  [draw opacity=0] (306,64.48) -- (326.5,64.48) -- (326.5,85.48) -- (306,85.48) -- cycle ; \draw   (306,64.48) -- (306,85.48)(326,64.48) -- (326,85.48) ; \draw   (306,64.48) -- (326.5,64.48)(306,84.48) -- (326.5,84.48) ; \draw    ;

% Text Node
\draw (82,91.48) node  [align=left] {$\displaystyle k_{2} :K\rightarrow \mathbb{Z}$};
% Text Node
\draw (210,113.48) node  [align=left] {$\displaystyle c_2 :C=C\cap S \rightarrow \mathbb{Z}$};
% Text Node
\draw (352,15.48) node  [align=left] {: Origin};
% Text Node
\draw (447,42.48) node  [align=left] {: Value of function is $\displaystyle x$ at the position \ };
% Text Node
\draw (415,73.48) node  [align=left] {: Position not in the domain};
% Text Node
\draw (315,42.48) node  [align=left] {$\displaystyle x$};
% Text Node
\draw (81,49.48) node  [align=left] {0};
% Text Node
\draw (101,68.48) node  [align=left] {10};
% Text Node
\draw (61,29.48) node  [align=left] {10};
% Text Node
\draw (81,29.48) node  [align=left] {10};
% Text Node
\draw (101,29.48) node  [align=left] {10};
% Text Node
\draw (61,49.48) node  [align=left] {10};
% Text Node
\draw (81,69.48) node  [align=left] {10};
% Text Node
\draw (61,68.48) node  [align=left] {10};
% Text Node
\draw (101,49.48) node  [align=left] {10};
% Text Node
\draw (171,33.48) node  [align=left] { };
% Text Node
\draw (191,53.48) node  [align=left] { };
% Text Node
\draw (211,73.48) node  [align=left] { };
% Text Node
\draw (231,93.48) node  [align=left] { };
% Text Node
\draw (171,53.48) node  [align=left] {10};
% Text Node
\draw (191,73.48) node  [align=left] { };
% Text Node
\draw (211,93.48) node  [align=left] {10};
% Text Node
\draw (191,93.48) node  [align=left] { };
% Text Node
\draw (231,73.48) node  [align=left] { };
% Text Node
\draw (191,13.48) node  [align=left] { };
% Text Node
\draw (211,33.48) node  [align=left] { };
% Text Node
\draw (231,53.48) node  [align=left] { };
% Text Node
\draw (251,73.48) node  [align=left] { };
% Text Node
\draw (171,73.48) node  [align=left] { };
% Text Node
\draw (191,33.48) node  [align=left] { };
% Text Node
\draw (211,13.48) node  [align=left] {10};
% Text Node
\draw (231,33.48) node  [align=left] { };
% Text Node
\draw (251,53.48) node  [align=left] {10};
% Text Node
\draw (231,13.48) node  [align=left] { };
% Text Node
\draw (251,33.48) node  [align=left] { };
% Text Node
\draw (211,53.48) node  [align=left] {10};
% Text Node
\draw (171,12.92) node  [align=left] {10};
% Text Node
\draw (250,12.92) node  [align=left] {10};
% Text Node
\draw (251,91.92) node  [align=left] {10};
% Text Node
\draw (172,91.92) node  [align=left] {10};

\end{tikzpicture}

\caption{The non-flat structuring elements $k_2$ and $c_2$ used in examples.}
\label{fig:kc-nonflat}
\end{figure*}
%%%%%%%%%%%%%%%%%%%%%%%%%%%%%%%%%%%%%%%%%%%%%%%%%%%%%%%%%%%%%%%%%%%%%%%%%%%%%%%%%%%%%%%%%%%%%%%%%%%%%%%%%%%

We observe that opening in sampled domain (here, $f\vert_S\circ b\vert_S $) is bounded above by sampling after opening of maximal reconstruction of the image (here, $((f\vert_S\oplus k)\circ b)\vert_S $) and bounded below by sampling after opening by maximal reconstruction of the SE (here, $f\circ [b\vert_S\oplus k])\vert_S $). Similarly, closing in sampled domain (here, $f\vert_S \bullet b\vert_S $) is bounded above by sampling after closing with maximal reconstruction of SE (here, $(f\bullet (b\vert_S \oplus k))\vert_S $) and bounded below by sampling after closing the minimal reconstruction (here, $((f\vert_S \bullet k)\bullet b)\vert_S$). 

In similar fashion, Figures \ref{fig:nonflatgrey411}-\ref{fig:nonflatgrey413} demonstrate the interaction of sampling with opening operation for non-flat structuring elements, and Figures  \ref{fig:nonflatgrey421}-\ref{fig:nonflatgrey423} demonstrate the interaction of sampling with closing operation for non-flat SEs.

%optional additional lines.
If the image coincides with its maximal or minimal reconstruction, then it satisfies some additional properties with respect to sampling and opening respectively closing. 
These are mentioned in Theorem \ref{thm:b5a}, which  directly follows Theorem \ref{thm:b4a}.

\begin{theorem}{Grey-value Sample Opening and Closing Theorem }
\label{thm:b5a}
If $B=B\circ K$, $b=b\circ k$, then
\begin{enumerate}[I.]
\item If $F=\: (F\cap S)\oplus K $, $f=f\vert_S \oplus k$ , $B=(B\cap S) \oplus K$ and $b=b\vert_S \oplus k$, then $f\vert_S \circ b\vert_S  =\: (f\circ b)\vert_S$.   
\item If  $F=(F\cap S)\bullet K$ , $f=f\vert_S\bullet k$, $B=(B\cap S) \oplus K$ and $b=b\vert_S \oplus k$, then $f\vert_S \bullet b\vert_S =\: (f\bullet b)\vert_S$.
\end{enumerate}
\end{theorem}

\section{Max-pooling and Reconstruction with Non-flat SEs}
\label{Sec:HeijExt}

The max-pooling operation introduced in \cite{ZC-88} is often used in CNNs \cite{Goodfellow-et-al-2016}. Max-pooling is a morphological operation, 
more precisely, it is morphological dilation by a square or rectangular flat filter followed by sampling \cite{Franchi-2020}.

The \textit{Sampling Operator}, $\sigma (.)$, as defined below, generalizes max-pooling to sampling after dilating with a \textit{paraboloid} filter $k$. In this section, we study about generalized max-pooling (i.e.\ $\sigma (.)$), its corresponding reconstruction and the effects of morphologically operating after max-pooling.

We assume that the sampling sieve $S\subseteq E^{N} $ and  the filter  $k:K\rightarrow L$ satisfies conditions \textit{I-VII} of Grey-value Digital Morphological Sampling Theorem \ref{thm:b6}. We extend the definitions and results of \cite{r3} to non-flat structuring element. That is, we \emph{do not} impose the condition $k(u)=0$, $\forall u \in K$ or $c(x)=0$, $\forall x\in C=C\cap S$.  In this section, we have used the non-flat SEs $k_2$ and $c_2$, as described in Figure \ref{fig:kc-nonflat}, for the examples.

\begin{definition} {Sampling Operator}
\label{def:hb1}
Let  $F\subseteq E^{N}$ and  $f:F\rightarrow L$ be the image. The structuring element  $k:K\rightarrow L$ and sieve $S$ are defined as above. The sampling operator is denoted by $\sigma (.)$ and is defined as \\ $(\sigma (f) )= (f\oplus k)\vert_S$, that is $(\sigma (f) )(s)= (f\oplus k)\vert_S (s)$, $\forall s\in (F\oplus K)\cap S$. 
\end{definition}

Similarly, reconstructing operator is defined. 

\begin{definition} {Reconstructing Operator}
\label{def:hb2}
Let  $G\subseteq S$ and  $g:G\rightarrow L$ be the sampled image. The structuring element  $k:K\rightarrow L$ and sieve $S$ are defined as above. The reconstructing operator is denoted by $\dot \sigma (.)$ and is defined as\\ $\dot \sigma (g) = g\bullet k$, that is, $(\dot \sigma (g))(x) = (g\bullet k)(x)$, $\forall x\in G\bullet K$. 
\end{definition}

Notice that \textit{Reconstructing Operator} uses morphological closing for reconstruction, i.e.\ minimal reconstruction, as described via Grey-value Digital Morphological Sampling Theorem \ref{thm:b6}. The choice of closing for reconstruction allows the \textit{Reconstructing Operator} to form an \emph{algebraic adjunction} with the \textit{Sampling Operator}. Here, $(\alpha (.),\beta(.) )$ is an algebraic adjunction if $\beta (f) \leq g$ iff $f \leq \alpha(g)$. 

We show that $(\dot \sigma , \sigma)$ forms an adjunction with the non-flat SE as well.

\begin{prop}
\label{thm:p1}
Let  $f:F(\subseteq E^{N})\rightarrow L$ be an image in the unsampled domain and  $g:G(\subseteq S)\rightarrow L$ be an image in sampled domain. \\Then, $\sigma(f) \leq g$ $\Leftrightarrow f\leq \dot \sigma(g)$.\\
i.e $(f\oplus k)\vert_S \leq g$ $\Leftrightarrow f\leq g\bullet k$. 
\end{prop}

\begin{proof}
Let $(f\oplus k)\vert_S \: \leq g$. Then $(f\oplus k)\vert_S \oplus k\: \leq g\oplus k$.  \\
By result \ref{res:g5} of Grey-value Digital Morphological Sampling Theorem \ref{thm:b6}, we have $f\leq f\vert_S \oplus k$. This gives $f\oplus k$ $\leq (f\oplus k)\vert_S \oplus k$ $\leq g\oplus k$. i.e $f\oplus k$ $\leq f\oplus k$. \\
Since grey-value erosion and and dilation forms an adjunction, using  \textit{Proposition 65} of \cite{r1}, we have, $f$ $\leq (g\oplus k)\ominus k$ $= g\bullet k$.\\
\\
 Conversely, let $f$ $\leq g\bullet k$. $G=G\cap S$, therefore $g=g\vert_S$.\\
 $f$ $\leq g\vert_S \bullet k$ $\Rightarrow (f\oplus k) \leq \: g\vert_S \oplus k$ by \textit{Proposition 65} of \cite{r1}. 
 \\ By result \text{II} of Grey-value Digital Morphological Sampling Theorem \ref{thm:b6}, we have, $(g\vert_S\oplus k)\vert_S$ $=g\vert_S$. 
 \\$\Rightarrow \: (f\oplus k)\vert_S$ $\leq (g\vert_S\oplus k)\vert_S$ $=g\vert_S$ $=g$.
\end{proof}

\begin{definition} {Reconstruction Operator}
\label{def:hb3}
Let  $F\subseteq E^{N}$ and  $f:F\rightarrow L$ be the image. The structuring element  $k:K\rightarrow L$ and sieve $S$ are defined as above. The reconstruction operator is denoted by $\rho (.)$, and it is defined by\\ $\rho (f) = \dot  \sigma (\sigma (f)) $ $=(f\oplus k)\vert_S\bullet k$.
\end{definition}

Similarly, an upper bound of \emph{Reconstruction operator} is given by $\delta(.)$ defined as \\$\delta(f)$ $=(f\oplus k)\vert_S \oplus k$.  Clearly, by Lemma \ref{thm:l3}, $(f\oplus k)\vert_S\bullet k$ $\leq (f\oplus k)\vert_S \oplus k$, i.e $\rho(f) \leq \: \delta(f)$, for any given image $f$.

Note that \emph{Reconstruction operator} ($\rho (.)$, Definition \ref{def:hb3}) is distinct from \emph{Reconstructing operator} ($\dot \sigma (.)$, Definition \ref{def:hb2}). \emph{Reconstruction operator} is used to study the effects in the image after a cycle of sampling (i.e.\ generalized max-pooling, with $\sigma (.)$) and reconstructing (with $\dot \sigma(.)$).

%%%%%%%%%%%%%%%%%%%%%%%%%%%%%%%%%%%%%%%%%%%%%%%%%%%%%%%%%%%%%%%%%%%%%%%%%%%%%%%%
%Figure F1
\begin{figure*}[!htbp]
	\includegraphics[width=0.3\linewidth]{Pictures//GreyValue/TestImageGreyValue.png}
\hfill
        \includegraphics[width=0.3\linewidth]{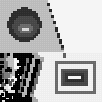}
 \hfill
 	\includegraphics[width=0.3\linewidth]{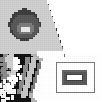}
        \vspace{1.0ex}
	\caption{Grey-value example image $f$ (left), its reconstruction
          $\rho(f)$ (centre) and its reconstruction $\delta(f)$ (right).  We can clearly notice, that $f\leq \rho(f) \leq \delta(f)$.  In this figure, we have used  non-flat SEs $k_2$ and $c_2$ as described in  \ref{fig:kc-nonflat}. }
	\label{fig:H2}
\end{figure*}
%%%%%%%%%%%%%%%%%%%%%%%%%%%%%%%%%%%%%%%%%%%%%%%%%%%%%%%%%%%%%%%%%%%%%%%%%%%%%%%%%%%

We have shown that $(\dot \sigma, \sigma)$ is an adjunction.  By \textit{Proposition 2.6} of \cite{r4} we have the following lemma.

\begin{lemma}
\label{thm:hl1}
Let $f:F(\subseteq E^{N})\rightarrow L$ be any image. Then, $f$ $\leq \rho (f)$ $=\dot \sigma (\sigma(f))$ $=(f\oplus k)\vert_S \bullet k$. 
\end{lemma}

Now, we give the relation between operating after max-pooling and max-pooling after performing the morphological operation, when $\sigma(.)$ and $\rho(.)$ is used for max-pooling and reconstruction, respectively. Let us reiterate. Both the filter $k:K\rightarrow L$ used in max-pooling and reconstruction operators and $c:C(\subseteq S)\rightarrow L$ are \emph{not restricted} to be flat. Therefore $k$, $K$ and $S$ must only satisfy the conditions of Grey-value Morphological Sampling Theorem \ref{thm:b6}, and $c$ can be any arbitrary non-negative function defined on in the sampled domain, 
i.e.\ $C=C\cap S$ and $c=c\vert _S$. Some minor results utilized in the proof of the following proposition are given in the Appendix, Section \ref{Sec:A}.

\begin{figure*}[!htbp]
\minipage{0.24\linewidth}
 \includegraphics[width=\linewidth]{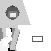}
  \caption{$\sigma(f\oplus c_2 )$}\label{fig:nonflatHeij11}
\endminipage \hspace{0.14\linewidth}
\minipage{0.24\linewidth}
  \includegraphics[width=\linewidth]{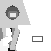}
  \caption{$ \sigma(f) \oplus c_2 $}\label{fig:nonflatHeij12}
\endminipage\hfill
\end{figure*}

\begin{figure*}[!htbp]
\minipage{0.24\linewidth}
 \includegraphics[width=\linewidth]{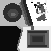}
  \caption{$\sigma(f\ominus c_2 )$}\label{fig:nonflatHeij21}
\endminipage\hfill
\minipage{0.24\linewidth}
  \includegraphics[width=\linewidth]{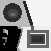}
  \caption{$ \sigma(f) \ominus c_2$}\label{fig:nonflatHeij22}
\endminipage\hfill
\minipage{0.24\linewidth}
  \includegraphics[width=\linewidth]{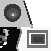}
  \caption{$\sigma((f\oplus k_2 )\ominus c_2 ) $}\label{fig:nonflatHeij23}
\endminipage\hfill
\end{figure*}

\begin{figure*}[!htbp]
\minipage{0.24\linewidth}
 \includegraphics[width=\linewidth]{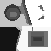}
  \caption{$\sigma(f\circ c_2 )$}\label{fig:nonflatHeij31}
\endminipage\hfill
\minipage{0.24\linewidth}
  \includegraphics[width=\linewidth]{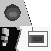}
  \caption{$ \sigma(f) \circ c_2 $}\label{fig:nonflatHeij32}
\endminipage\hfill
\minipage{0.24\linewidth}
  \includegraphics[width=\linewidth]{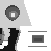}
  \caption{$\sigma((f\oplus k_2 )\circ c_2 ) $}\label{fig:nonflatHeij33}
\endminipage\hfill
\end{figure*}

\begin{figure*}[!htbp]
\minipage{0.24\linewidth}
 \includegraphics[width=\linewidth]{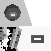}
  \caption{$\sigma(f\bullet c_2 )$}\label{fig:nonflatHeij41}
\endminipage\hfill
\minipage{0.24\linewidth}
  \includegraphics[width=\linewidth]{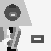}
  \caption{$ \sigma(f) \bullet c_2 $}\label{fig:nonflatHeij42}
\endminipage\hfill
\minipage{0.24\linewidth}
  \includegraphics[width=\linewidth]{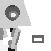}
  \caption{$\sigma((f\oplus k_2 )\bullet c_2 ) $}\label{fig:nonflatHeij43}
\endminipage\hfill
\end{figure*}

\begin{figure*}[!htbp]
\minipage{0.3\linewidth}
 \includegraphics[width=\linewidth]{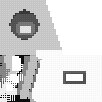}
  \caption{$\rho(f)\oplus c_2 $}\label{fig:nonflatHeij51}
\endminipage \hspace{0.14\linewidth}
\minipage{0.3\linewidth}
  \includegraphics[width=\linewidth]{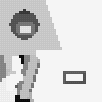}
  \caption{$ \rho(f\oplus c_2 )$}\label{fig:nonflatHeij52}
\endminipage\hfill
\end{figure*}

\begin{figure*}[!htbp]
\minipage{0.3\linewidth}
 \includegraphics[width=\linewidth]{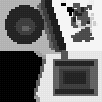}
  \caption{$\rho(f\ominus c_2 )$}\label{fig:nonflatHeij61}
\endminipage  \hspace{0.14\linewidth}
\minipage{0.3\linewidth}
  \includegraphics[width=\linewidth]{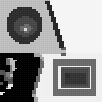}
  \caption{$ \rho((f\oplus k_2 )\ominus c_2 ) $}\label{fig:nonflatHeij63}
\endminipage\hfill
\end{figure*}

\begin{figure*}[!htbp]
\minipage{0.3\linewidth}
 \includegraphics[width=\linewidth]{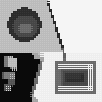}
  \caption{$\rho(f)\circ c_2 $}\label{fig:nonflatHeij71}
\endminipage \hspace{0.14\linewidth}
\minipage{0.3\linewidth}
  \includegraphics[width=\linewidth]{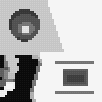}
  \caption{$  \rho((f\oplus k_2 )\circ c_2 )$}\label{fig:nonflatHeij72}
\endminipage\hfill
\end{figure*}

\begin{figure*}[!htbp] 
\minipage{0.3\linewidth}
 \includegraphics[width=\linewidth]{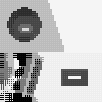}
  \caption{$\rho(f)\bullet c_2 $}\label{fig:nonflatHeij81}
\endminipage \hspace{0.14\linewidth}
\minipage{0.3\linewidth}
  \includegraphics[width=\linewidth]{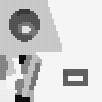}
  \caption{$  \rho((f\oplus k_2 )\bullet c_2 )$}\label{fig:nonflatHeij82}
\endminipage\hfill
\end{figure*}

\begin{prop}
\label{thm:h2}
Let  $F\subseteq E^{N}$ and  $f:F\rightarrow L$ be the image.  The structuring elements $c:C(\subseteq S)\rightarrow L$, $k:K\rightarrow L$ and sieve $S$ are defined as above. Then, 
\begin{enumerate}[I.]
\item $\sigma(f\oplus c)$ $= \sigma(f) \oplus c$,\\
 i.e.\ $((f\oplus c)\oplus k)\vert_S$ $=(f\oplus k)\vert_S \oplus c$

\item $\sigma (f\ominus c)$ $\leq \sigma(f)\ominus c$ $\sigma((f\oplus k)\ominus c)$,\\
i.e.\ $((f\ominus c)\oplus k)\vert_S$ $\leq (f\oplus k)\vert_S \ominus c$ $\leq (((f\oplus k)\ominus c)\oplus k)\vert_S$

\item $\sigma(f\circ c)$ $\leq \sigma(f)\circ c$ $\leq \sigma((f\oplus k)\circ c)$,\\
i.e.\ $((f\circ c)\oplus k)\vert_S$ $\leq (f\oplus k)\vert_S \circ c$ $\leq (((f\oplus k)\circ c)\oplus k)\vert_S$

\item $\sigma(f\bullet c)$ $\leq \sigma(f) \bullet c$ $\leq \sigma ((f\oplus k)\bullet c)$,\\
i.e.\ $((f\bullet c)\oplus k)\vert_S$ $\leq (f\oplus k)\vert_S \bullet c$ $\leq (((f\oplus k)\bullet c)\oplus k)\vert_S$

\item $\rho(f)\oplus c$ $\leq \rho(f\oplus c)$,\\
i.e.\ $((f\oplus k)\vert_S\bullet k)\oplus c$ $\leq ((f\oplus c)\oplus k)\vert_S\bullet k$

\item $\rho(f\ominus c)$ $\leq \rho (f) \ominus c$ $\leq \rho((f\oplus k)\ominus c)$,\\
i.e.\ $((f\ominus c)\oplus k)\vert_S$ $\leq ((f\oplus k)\vert_S \bullet k)\ominus c$ $\leq(((f\oplus k)\ominus c)\oplus k)\vert_S\bullet k $

\item $\rho(f) \circ c$ $\leq \rho((f\oplus k)\circ c)$,\\
i.e.\ $((f\oplus k)\vert_S \bullet k)\circ c$ $\leq (((f\oplus k)\circ c)\oplus k)\vert_S \bullet k$

\item $\rho(f)\bullet c$ $\leq \rho((f\oplus k)\bullet c)$,\\
i.e.\ $((f\oplus k)\vert_S \bullet k)\bullet c$ $\leq (((f\oplus k)\bullet c)\oplus k)\vert_S \bullet k$

\end{enumerate}
\end{prop}

\begin{proof}
\textit{I.}
\begin{align*}
((f\oplus c)\oplus k)\vert_S &= ((f\oplus k)\oplus c)\vert_S  \text{,  by \textit{Proposition 60} of \cite{r1}}\\
&=((f\oplus k)\oplus c\vert_S)\vert_S \text{,  because %$C=C\cap S$,
$c=c\vert_S$}\\
&=(f\oplus k)\vert_S \oplus c\vert_S  \text{,  by Lemma \ref{thm:lemaa}, \textit{I}}\\
&=(f\oplus k)\vert_S \oplus c
\end{align*}

\textit{II.}
We first show $((f\ominus c)\oplus k)\vert_S$ $\leq (f\oplus k)\vert_S \ominus c$.
\begin{align*}
((f\ominus c)\oplus k)\vert_S &\leq ((f\oplus k)\ominus c)\vert_S \text{,  by Lemma \ref{thm:l4}, \textit{II}}\\
&=((f\oplus k)\ominus c\vert_S)\vert_S \text{,  $\because$ %$C=C\cap S$,
$c=c\vert_S$}\\
&=(f\oplus k)\vert_S \ominus c\vert_S \text{,  by Lemma  \ref{thm:lemaa}, \textit{II}}\\
&=(f\oplus k)\vert_S \ominus c
\end{align*}

We now show $(f\oplus k)\vert_S \ominus c$ $\leq (((f\oplus k)\ominus c)\oplus k)\vert_S$.
\begin{align*}
(f\oplus k)\vert_S \ominus c &=(f\oplus k)\vert_S \ominus c\vert_S \text{,  because %$C=C\cap S$,
$c=c\vert_S$}\\
&=((f\oplus k)\ominus c\vert_S)\vert_S \text{,  by  Lemma \ref{thm:lemaa}, \textit{II}}\\
&=((f\oplus k)\ominus c)\vert_S \\
&\leq (((f\oplus k)\ominus c)\oplus k)\vert_S \text{,  by Lemma \ref{thm:l1}}\\
\end{align*}

\textit{III.}
We first show $((f\circ c)\oplus k)\vert_S$ $\leq (f\oplus k)\vert_S \circ c$. 
\begin{align*}
((f\circ c)\oplus k)\vert_S &\leq ((f\oplus k)\circ c)\vert_S   \text{,  by Lemma \ref{thm:l5}, \textit{II}}\\
&=((f\oplus k)\circ c\vert_S )\vert_S  \text{,  because %$C=C\cap S$,
$c=c\vert_S$}\\
&=(f\oplus k)\vert_S \circ c\vert_S \text{,  by Proposition \ref{thm:pb16a}}\\
&=(f\oplus k)\vert_S \circ c
\end{align*}

We now show  $(f\oplus k)\vert_S \circ c$ $\leq (((f\oplus k)\circ c)\oplus k)\vert_S$.
\begin{align*}
(f\oplus k)\vert_S \circ c &= (f\oplus k)\vert_S \circ c\vert_S \text{,  because % $C=C\cap S$,
$c=c\vert_S$}\\
&=((f\oplus k)\circ c\vert_S )\vert_S \text{,  by Proposition \ref{thm:pb16a}}\\
&=((f\oplus k)\circ c )\vert_S\\
&\leq (((f\oplus k)\circ c )\oplus k)\vert_S \text{,  by Lemma \ref{thm:l1}}
\end{align*}

\textit{IV.}
We first show $((f\bullet c)\oplus k)\vert_S$ $\leq (f\oplus k)\vert_S \bullet c$. 
\begin{align*}
((f\bullet c)\oplus k)\vert_S &\leq ((f\oplus k)\bullet c)\vert_S   \text{,  by Lemma \ref{thm:l5}, \textit{II}}\\
&=((f\oplus k)\bullet c\vert_S )\vert_S  \text{,  because %$C=C\cap S$,
$c=c\vert_S$}\\
&=(f\oplus k)\vert_S \bullet c\vert_S \text{,  by Proposition \ref{thm:pb17a}}\\
&=(f\oplus k)\vert_S \bullet c
\end{align*}

We now show  $(f\oplus k)\vert_S \bullet c$ $\leq (((f\oplus k)\bullet c)\oplus k)\vert_S$.
\begin{align*}
(f\oplus k)\vert_S \bullet c &= (f\oplus k)\vert_S \bullet c\vert_S \text{,  because %$C=C\cap S$,
$c=c\vert_S$}\\
&=((f\oplus k)\bullet c\vert_S )\vert_S \text{,  by Proposition \ref{thm:pb17a}}\\
&=((f\oplus k)\bullet c )\vert_S\\
&\leq (((f\oplus k)\bullet c )\oplus k)\vert_S \text{,  by Lemma \ref{thm:l1}}
\end{align*}

\textit{V.}
\begin{align*}
((f\oplus k)\vert_S\bullet k)\oplus c &\leq ((f\oplus k)\vert_S\oplus c)\bullet k, \\
& \text{  by Lemma \ref{thm:l6}, \textit{II}}\\
&=((f\oplus k)\vert_S\oplus c\vert_S)\bullet k, \\
&\text{  because %$C=C\cap S$,
$c=c\vert_S$}\\
&=((f\oplus k)\oplus c\vert_S)\vert_S\bullet k, \\
& \text{  by Lemma \ref{thm:lemaa}, \textit{I} }\\
&=((f\oplus k)\oplus c)\vert_S\bullet k\\
&=((f\oplus c)\oplus k)\vert_S\bullet k, \\
&\text{ by \textit{Proposition 60} of \cite{r1} } 
\end{align*}

\textit{VI.}
We first show $((f\ominus c)\oplus k)\vert_S$ $\leq ((f\oplus k)\vert_S \bullet k)\ominus c$.
\begin{align*}
((f\ominus c)\oplus k)\vert_S \bullet k &\leq ((f\oplus k)\ominus c)\vert_S\bullet k, 
\\
& \text{,  by Lemma \ref{thm:l4}  \textit{II}}\\
&=((f\oplus k)\ominus c\vert_S)\vert_S\bullet k , \\
& \text{  because %$C=C\cap S$,
$c=c\vert_S$}\\
&=((f\oplus k)\vert_S\ominus c\vert_S)\vert\bullet k , \\
&
\text{ by Lemma \ref{thm:lemaa}, \textit{II}}\\
&=((f\oplus k)\vert_S\ominus c)\vert\bullet k  \\
&\leq ((f\oplus k)\vert_S \bullet k)\ominus c, \\
&\text{  by Lemma \ref{thm:l7}}
\end{align*}

In the remaining parts we apply the result \textit{III} of Grey-value Morphological Sampling Theorem, \ref{thm:b6}, to $(f\oplus k)$ to obtain 
\begin{equation} 
\label{eq:1}
(f\oplus k)\vert_S \bullet k \leq (f\oplus k). \tag{Eq1}
\end{equation}
\\
We now show $((f\oplus k)\vert_S \bullet k)\ominus c$ $\leq(((f\oplus k)\ominus c)\oplus k)\vert_S\bullet k $.
\begin{align*}
((f\oplus k)\vert_S \bullet k)\ominus c &\leq  (f\oplus k)\ominus c \text{,  by \eqref{eq:1}}\\
&\leq (((f\oplus k)\ominus c)\oplus k)\vert_S\bullet k,\\
&\text{  by Lemma \ref{thm:hl1}}\\
\end{align*}

\textit{VII.}
\begin{align*}
((f\oplus k)\vert_S \bullet k)\circ c &\leq  (f\oplus k)\circ c \text{,  by \eqref{eq:1}}\\
&\leq  (((f\oplus k)\circ c)\oplus k)\vert_S \bullet k , \\
&\text{  by Lemma \ref{thm:hl1}}\\
\end{align*}

\textit{VIII.}
\begin{align*}
((f\oplus k)\vert_S \bullet k)\bullet c &\leq  (f\oplus k)\bullet c \text{,  by \eqref{eq:1}}\\
&\leq  (((f\oplus k)\bullet c)\oplus k)\vert_S \bullet k , \\
\text{  by Lemma \ref{thm:hl1}}\\
\end{align*}
%Edit This Out ABCd
\end{proof}

The results \textit{I -IV} give relations in sampled domain. The results \textit{V-VIII} give relations in reconstructed images.  The interaction between sampling operator $\sigma(.)$ and morphological operations with non-flat SEs is illustrated in Figures \ref{fig:nonflatHeij11}-\ref{fig:nonflatHeij43}. Figures \ref{fig:nonflatHeij51} to \ref{fig:nonflatHeij82} illustrate interaction of reconstruction operation $\rho(.)$ with morphological operations using non-flat SEs. In the following examples we have used the non-flat SEs $k_2$ and $c_2$ as described in Figure \ref{fig:kc-nonflat}.

\section{Conclusion}

In this paper we have built upon classic works of Haralick and co-authors, and we have shown 
in detail how to transfer digital sampling theorems concerned with four fundamental morphological operations, namely dilation, erosion, opening and closing, from the binary setting to grey-value images.

Using the above results, we have also extended the work of Heijmans and Toet on max-pooling,  morphological sampling and reconstruction to use of 
non-flat structuring elements. 

With this paper we have not only worked on closing a gap in the foundation 
of mathematical morphology. Our aim is also 
to address some fundamental theoretical aspects of the pooling operation 
encountered in machine learning. In our future work we also strive to 
elaborate more on this aspect.
\begin{appendices}

\section{Some Minor Results} %Done Formatting
\label{Sec:A}
In this section, we present some smaller results which are utilized in the discussions of Section \ref{Sec:HeijExt}. 

Let $F,\: G,\: C,\: K \subseteq E^{N} $ and $f:F\rightarrow L$, $g:G\rightarrow L$, $c:C\rightarrow L$  and $k:K\rightarrow L$. Let $S\subseteq E^{N}$ be the sampling sieve. 

\begin{lemma}
\label{thm:l1}
Let $0\in K$, $k(0) =0$ and $ k(u)\geq 0$ $\forall u\in K$. Then,  $f\leq f\oplus k$.
\end{lemma}

\begin{proof}
We first show that if $0\in K$ then $F\subseteq F\oplus K$.
Clearly,  $x\in F$ and $0\in K$ $\Rightarrow$ $(x+0) \in F\oplus K$. Thus, $x\in F\oplus K$. $\therefore F\subseteq F\oplus K$.\\
We now show  $f\leq f\oplus k$. For all $x\in F$ $\subseteq F\oplus K$, we have 
$(f\oplus k)(x)$  $= \max _{x-u\in F,\: u\in K} \{f(x-u) + \:k(u)\} $ $\geq f(x-0)+k(0) $ $=f(x)$\\
$\therefore$ $\forall x\in F$, $ f(x)\leq (f\oplus k)(x)$.
\end{proof}

\begin{lemma}
\label{thm:l2}
If $f\leq g$ then $f\vert _S \leq g\vert_S$.
\end{lemma}

\begin{proof}
We have, $\forall x\in F$ $\subseteq G$, $f(x)\leq g(x)$. That is $\forall s \in F\cap S \subseteq G\cap S$, $f\vert_S (s) \leq g\vert_S (s)$. Therefore, $f\vert_S \leq g\vert_S$. 
\end{proof} 

\begin{lemma}
\label{thm:l3}
Let $0\in K$, $k(0) =0$ and $ k(u)\geq 0$ $\forall u\in K$. Then, $f\ominus k \leq f$.
\end{lemma}

\begin{proof}
We first show that if $0\in K$ then $F\ominus K \subseteq F$.\\
Clearly, $x\in F\ominus K$ $\Rightarrow \forall u \in K$, $x+u \in F$.\\ 
$0 \in K$  $\therefore x+0 =x \in F$, which implies $F\ominus K \subseteq F$.\\   
We now show $f\ominus k \leq f$. \\
For any $x\in F\ominus K$, we have, $(f\ominus k)(x)$ $=\min _{x+u\in F, \: u\in K} \{f(x+u)-k(u)\}$ $\leq f(x+0)-k(0)$ $=f(x)$. 
\end{proof}

\begin{lemma}
\label{thm:l4}
\begin{enumerate}[I.]
\item $(F\ominus C)\oplus K \subseteq (F\oplus K)\ominus C$
\item $(f\ominus c)\oplus k \leq (f \oplus k)\ominus c$
\end{enumerate}
\end{lemma}

\begin{proof}
\textit{I.}
\\
Let $x\in (F\ominus C)\oplus K$. Then, $x=u+k$ where $u\in F\ominus C$, $k\in K$.  \\
$u\in F\ominus C$ implies $\forall c_0 \in C$, $u+c_0 \in F$. \\
$\Rightarrow $ $\forall c_0 \in C$, $(u+k)+c_0$ $=(u+c_0)+k \in F\oplus K$.\\
Therefore, $x=u+k$ $\in (F\oplus K)\ominus C$.
\\
\textit{II.}
\\ %verify proof 
From Umbra Homomorphism Theorem (\textit{Theorem 58} of \cite{r1} ), we have $U[(f\ominus c)\oplus k] $ $=U[f\ominus c]\oplus U[k]$ $=(U[f]\ominus U[c])\oplus U[k]$. Similarly, $U[(f\oplus k)\ominus c] =$ $(U[f]\oplus U[k])\ominus U[c] $.\\ 
We have $(F\ominus C)\oplus K \subseteq (F\oplus K)\ominus C$. 
Also, using the logic of Part \textit{I} of the proof, $(U[f]\ominus U[c])\oplus U[k] \subseteq  (U[f]\oplus U[k])\ominus U[c] $.\\
Therefore, we can conclude,  $(f\ominus c)\oplus k \leq (f \oplus k)\ominus c$.
\end{proof}

\begin{lemma}
\label{thm:l5}
\begin{enumerate}[I.]
\item $(F\circ C)\oplus K$ $\subseteq (F\oplus K)\circ C$
\item $(f\circ c)\oplus k$ $\leq (f\oplus k)\circ c$ 
\end{enumerate}
\end{lemma}

\begin{proof}
\textit{I.}\\
\begin{align*}
(F\circ C)\oplus K &= ((F\ominus C)\oplus C)\oplus K\\
&=((F\ominus C)\oplus K\oplus C \text{,  by \textit{Prop.\ 3} of \cite{r1}}\\
&\subseteq ((F\oplus K)\ominus C)\oplus C \text{,  by Lemma \ref{thm:l4}, \textit{I}}\\
&=(F\oplus K)\circ C
\end{align*}
\textit{II.}\\
\begin{align*}
(f\circ c)\oplus k &= ((f\ominus c)\oplus c)\oplus k\\
&= ((f\ominus c)\oplus k)\oplus c \text{,  by \textit{Prop.\ 60} of \cite{r1}}\\
&\leq ((f\oplus k)\ominus c)\oplus c  \text{,  by Lemma \ref{thm:l4}, \textit{II}}\\
&=(f\oplus k)\circ c
\end{align*}
\end{proof}

\begin{lemma}
\label{thm:l6}
\begin{enumerate}[I.]
\item  $(F\bullet C)\oplus K$ $\subseteq (F\oplus K)\bullet C$
\item  $(f\bullet c)\oplus k$ $\leq (f\oplus k)\bullet c$ 
\end{enumerate}
\end{lemma}

\begin{proof}
\textit{I.}\\
\begin{align*}
(F\bullet C)\oplus K &= ((F\oplus C)\ominus C)\oplus K\\
& \subseteq ((F\oplus C)\oplus K)\ominus C  \text{,  by Lemma \ref{thm:l4}, \textit{I}}\\
& =(F\oplus K)\bullet C
\end{align*}
\textit{II.}\\
\begin{align*}
(f\bullet c)\oplus k &= ((f\oplus c)\ominus c)\oplus k\\
& \leq ((f\oplus c)\oplus k)\ominus c  \text{,  by Lemma \ref{thm:l4}, \textit{II}}\\
&=((f\oplus k)\oplus c)\ominus c   \text{,  by \textit{Prop.\ 60} of \cite{r1}}\\
\end{align*}
\end{proof}

\begin{lemma}
\label{thm:l7}
$(f\ominus c)\bullet k$ $\leq (f\bullet k)\ominus c$
\end{lemma}

\begin{proof}
\begin{align*}
(f\ominus c)\bullet k & = ((f\ominus c)\oplus k)\ominus k \\
& \leq ((f\oplus k)\ominus c)\oplus k \text{,  by Lemma \ref{thm:l4}, \textit{II}}\\
& = (f\oplus k)\ominus (c\oplus k) \text{,  by \textit{Prop.\ 61} of \cite{r1}}\\
& = (f\oplus k) \ominus (k\oplus c) \text{,  by \textit{Prop.\ 59} of \cite{r1}}\\
&= ((f\oplus k)\ominus k)\ominus c \text{,  by \textit{Prop.\ 61} of \cite{r1}}\\
& = (f\bullet k)\ominus c
\end{align*}
\end{proof}




\end{appendices}



\begin{thebibliography}{99}

\bibitem{Serra-Soille}J.~Serra, and P.~Soille (Eds.): Mathematical morphology and its applications to image processing. Vol. 2. Springer Science and Business Media, (2012)

\bibitem{Najman-Talbot}L.~Najman, and H.~Talbot (Eds.): Mathematical Morphology: From Theory to Applications. Wiley-ISTE, (2010)

\bibitem{Roerdink-2011}J.B.T.M.~Roerdink, "Mathematical Morphology in Computer Graphics, Scientific Visualization and Visual Exploration", in Proc.~ISMM, LNCS {\bf 6671}, 367--380, (2011)

\bibitem{r1}R. M. Haralick, S. R. Sternberg, and X. Zhuang. Image analysis using mathematical morphology. IEEE Trans. Pattern Anal. Mach. Intell., 9(4):532–550, 1987.

\bibitem{r2}R. M. Haralick, X. Zhuang, C. Lin, and J. S. J. Lee. The digital morphological sampling theorem.
IEEE Transactions on Acoustics, Speech, and Signal Processing, 37(12):2067–2090, Dec 1989.

\bibitem{r3}H. Heijmans and A. Toet. Morphological sampling. CVGIP: Image Under- standing, 54(3):384 – 400, 1991.

\bibitem{r4}H. Heijmans and C. Ronse. The algebraic basis of mathematical morphology i. dilations and erosions. Computer Vision, Graphics, and Image Processing, 50(3):245 – 295, 1990.

%Morphological Neural Networks
\bibitem{Shen-2019} Shen, Yucong, Xin Zhong, and Frank Y. Shih. "Deep morphological neural networks." arXiv preprint arXiv:1909.01532 (2019).

\bibitem{Franchi-2020} Franchi, Gianni, Amin Fehri, and Angela Yao. "Deep morphological networks." Pattern Recognition 102 (2020): 107246.

\bibitem{Kirszenberg-2021} Kirszenberg, Alexandre, Guillaume Tochon, Élodie Puybareau, and Jesus Angulo. "Going beyond p-convolutions to learn grayscale morphological operators." In International Conference on Discrete Geometry and Mathematical Morphology, pp. 470-482. Springer, Cham, 2021.

\bibitem{Hermary-2022} Hermary, R., Tochon, G., Puybareau, É. et al. Learning Grayscale Mathematical Morphology with Smooth Morphological Layers. J Math Imaging Vis 64, 736–753 (2022).

\bibitem{r5}C. Ronse, H.J.A.M. Heijmans, The algebraic basis of mathematical morphology: II. Openings and closings, CVGIP: Image Understanding, Volume 54, Issue 1, 1991, Pages 74-97, ISSN 1049-9660.

\bibitem{r6}C. Ronse. 1990. Why mathematical morphology needs complete lattices. Signal Process. 21, 2 (October 1990), 129-154. 

\bibitem{r7}Henk J.A.M. Heijmans, A note on the umbra transform in grayscale morphology, Pattern Recognition Letters, Volume 14, Issue 11, 1993, Pages 877-881, ISSN 0167-8655,
\bibitem{Goodfellow-et-al-2016}Goodfellow,I., Bengio,Y., Courville,A.: Deep Learning. MIT Press (2016).
\bibitem{Haralick88}Haralick, R., Zhuang, X., Lin, C., Lee, J.S.C.: Binary morphology: working in the sampled domain. In: Proc. CVPR. pp. 780–791. IEEE Computer Society (1988)
\bibitem{Shannon-1998} Communication in the presence of noise. Proc. IEEE 86 (2) (1998).
\bibitem{ZC-88} Zhou,Y., Chellappa,R.: Computation of optical flow using a neural network. In:Proc. IEEE International Conference on Neural Networks. pp. 71-78. IEEE Computer Society (1988)
\bibitem{Matheron-75} Matheron,G.: Random Sets and Integral Geometry. Wiley, NewYork(1975)
\bibitem{Soille-2003}Soille,P.:Morphological Image Analysis. Springer, Berlin(2003)
\bibitem{Serra-82}Serra, J.: Image Analysis and Mathematical Morphology, vol. 1. Academic Press, London (1982)
\bibitem{Serra-88} Serra, J.: Image Analysis and Mathematical Morphology, vol. 2. Academic Press, London (1988)

\bibitem{sridhar-breuss-sampling-caip}
Sridhar, V., Breu{\ss}, M.:
Sampling of Non-flat Morphology for Grey Value Images.
In: Computer Analysis of Images and Patterns, Springer International Publishing, 2021, LNCS Vol. 13053, pp. 88-97
\end{thebibliography}


\end{document}